\documentclass[
10pt,
twocolumn,
superscriptaddress,
nofootinbib,
notitlepage,
amsmath,amssymb,
aps,physrev,
]{revtex4-2}

\usepackage{booktabs}

\usepackage[english]{babel}
\usepackage[T1]{fontenc}
\usepackage{amsmath}
\usepackage[hidelinks]{hyperref}
\usepackage[capitalize]{cleveref}
\usepackage{amssymb}
\usepackage{amsthm}
\usepackage{enumitem}
\usepackage[labelfont=bf]{caption}
\usepackage{pifont}

\newtheorem{theorem}{Theorem}
\newtheorem{definition}{Definition}
\newtheorem{lemma}{Lemma}

\newtheorem{proposition}{Proposition}
\newtheorem{remark}{Remark}

\theoremstyle{definition}
\newtheorem{example}{Example}

\crefname{theorem}{Theorem}{Theorems}
\Crefname{theorem}{Theorem}{Theorems}

\crefname{definition}{Definition}{Definitions}
\Crefname{definition}{Definition}{Definitions}

\crefname{lemma}{Lemma}{Lemmas}
\Crefname{lemma}{Lemma}{Lemmas}

\crefname{corollary}{Corollary}{Corollaries}
\Crefname{corollary}{Corollary}{Corollaries}

\crefname{proposition}{Proposition}{Propositions}
\Crefname{proposition}{Proposition}{Propositions}

\crefname{remark}{Remark}{Remarks}
\Crefname{remark}{Remark}{Remarks}

\crefname{problem}{Problem}{Problems}
\Crefname{problem}{Problem}{Problems}

\crefname{assumption}{Assumption}{Assumptions}
\Crefname{assumption}{Assumption}{Assumptions}

\crefname{example}{Example}{Examples}
\Crefname{example}{Example}{Examples}

\crefname{algocf}{Algorithm}{Algorithms}
\Crefname{algocf}{Algorithm}{Algorithms}
\crefname{algocfline}{Algorithm}{Algorithms}
\Crefname{algocfline}{Algorithm}{Algorithms}

\usepackage{subcaption}
\usepackage[ruled,vlined,linesnumbered,lined]{algorithm2e}
\usepackage[braket, qm]{qcircuit}
\usepackage{graphicx}
\usepackage{caption}
\usepackage{array}
\usepackage{enumitem}
\usepackage{makecell}

\newcommand{\R}{\mathbb{R}}

\newcommand{\Skew}{\operatorname{Skew}}

\newcommand{\grad}{\operatorname{grad}}

\begin{document}

\title{Achieving double-logarithmic precision dependence in optimization-based quantum unstructured search}

\author{Zhijian Lai}

\affiliation{Beijing International Center for Mathematical Research, Peking University, Beijing, China}

\author{Dong An}

\affiliation{Beijing International Center for Mathematical Research, Peking University, Beijing, China}

\author{Jiang Hu}

\affiliation{Yau Mathematical Sciences Center, Tsinghua University, Beijing, China}

\author{Zaiwen Wen}

\affiliation{Beijing International Center for Mathematical Research, Peking University, Beijing, China}

\date{\today}

\begin{abstract}
Grover's algorithm is a fundamental quantum algorithm that achieves a quadratic speedup for unstructured search problems of size $N$.
Recent studies have reformulated this task as a maximization problem on the unitary manifold and solved it via linearly convergent Riemannian gradient ascent (RGA) methods, resulting in a complexity of $\mathcal{O}(\sqrt{N/M}\log (1/\varepsilon))$, where \(M\) denotes the number of target items and \(\varepsilon\) denotes the success probability error.
In this work, we adopt the Riemannian modified Newton (RMN) method to solve the quantum search problem, under the assumption that the ratio $ M/N$ is known. We show that, in this setting, the Riemannian Newton direction is collinear with the Riemannian gradient in the sense that the Riemannian gradient is always an eigenvector of the corresponding Riemannian Hessian.
This structure removes the overhead of Hessian inversion and allows the proposed RMN method to retain the local quadratic convergence in terms of the error $\varepsilon$.
More precisely, we rigorously prove an overall complexity of $\mathcal{O}(\sqrt{N/M}+\log\log(1/\varepsilon) )$.
Furthermore, our approach remains Grover-compatible, namely, it relies exclusively on the standard Grover diffusion and oracle operators to ensure algorithmic implementability, and its parameter update process can be efficiently precomputed on classical computers.
\end{abstract}

\maketitle

\section{Introduction}

Unstructured search, the task of finding specific target elements in an unsorted space of size $N$, is a ubiquitous challenge across cryptographic analysis \cite{grassl2016applying,bernstein2025post}, optimization \cite{durr1996quantum,gilliam2021grover}, and machine learning \cite{dong2008quantum,du2021grover}.
Let $M\ll N$ denote the number of target items.
Classically, this problem is inherently inefficient, as guaranteeing success requires $\Omega(N/M)$ queries \cite{bennett1997strengths}. This constitutes a severe computational bottleneck for large-scale data retrieval \cite{knuth1997art,manning2008introduction}.

Grover's quantum search algorithm resolved this bottleneck by achieving an optimal query complexity of $\Theta(\sqrt{N/M})$ \cite{grover1996fast,nielsen2010quantum,zalka1999grover,beals2001quantum}, marking a landmark demonstration of quantum advantage.
Beyond its direct application to the search problem, Grover's mechanism is a fundamental primitive in quantum computing. Its principles form the basis of amplitude amplification \cite{bhmt2002amplitude,suzuki2020amplitude} and have been further generalized within the broader framework of quantum singular value transformation (QSVT) \cite{gilyen2019quantum,martyn2021grand}.

To gain deeper insights into the mechanics of Grover's quadratic speedup, researchers have explored various mathematical interpretations. Early efforts characterized Grover's iteration as a geodesic in complex projective space \cite{miyake2001geometric}. Other perspectives have used information geometry to describe the algorithm as a path under the Wigner-Yanase metric \cite{cafaro2012grover}, while recent work has proposed a connection to imaginary-time evolution (ITE) \cite{suzuki2025grover}. These interpretations link the efficiency of quantum search to the geometry of the quantum state space.

A promising modern approach to analyzing such quantum processes is the \textit{(Riemannian) manifold optimization} framework \cite{absil2008optimization,hu2020brief,boumal2023introduction}.
Since quantum operations are constrained by unitarity, they naturally lie on the unitary manifold, namely, the unitary group $\{U \in \mathbb{C}^{N \times N} | U^{\dagger} U = I\}$ \cite{hall2015liegroups}.
Formulating various quantum tasks as optimization problems on these Riemannian manifolds has led to significant progress in quantum tomography \cite{hsu2024quantum,li2025quantum}, ground-state preparation \cite{wiersema2023optimizing,magann2023randomized,malvetti2024randomized,pervez2025riemannian,lai2026quantum}, Hamiltonian simulation \cite{kotil2024riemannian,le2025riemannian}, and circuit compilation \cite{luchnikov2021qgopt,luchnikov2021riemannian,yao2024riemannian}.

Specifically, a recent study \cite{lai2025grover} reformulates Grover's search as a maximization problem on the unitary manifold and proposes a physically implementable \textit{Riemannian gradient ascent} (RGA) method to achieve the desired quadratic speedup in size $N$.
While this first-order approach establishes a rigorous link between the optimization framework and quantum dynamics, it is limited to linear convergence with respect to the accuracy parameter $\varepsilon$, with total complexity $\mathcal{O}(\sqrt{N/M}\log(1/\varepsilon))$.

To further accelerate convergence in terms of the target accuracy $\varepsilon$, second-order optimization methods, such as the Riemannian Newton method \cite{adler2002newton,absil2008optimization}, provide a natural approach.
In this work, we propose a \textit{Riemannian modified Newton} (RMN) method for Grover's search problem, where each iteration admits a physically implementable second-order update using only the diffusion and oracle operators, analogous to Grover's algorithm.

Typically, moving from first-order methods to second-order methods trades faster convergence for higher computational cost per iteration, often due to Hessian inversion.
A key theoretical result of our RMN method is that the Riemannian gradient of our cost function in unstructured search is always an eigenvector of the associated Riemannian Hessian.
This implies that the Newton direction reduces to a scaled gradient, eliminating the need for explicit Hessian inversion. Consequently, second-order updates incur no additional overhead compared to first-order updates.

We test our RMN method numerically and find that, once the iterates get close to the target state, RMN can attain very high-precision solutions in very few iterations.
Beyond this numerical evidence, we prove that \cref{alg-newton-armijo} enters a Newton region within \(\mathcal O(\sqrt{N/M})\) iterations and subsequently converges quadratically, which yields the total complexity \(\mathcal O(\sqrt{N/M}+\log\log(1/\varepsilon))\). The proof is provided separately in \cref{app-complexity}.

However, we emphasize that, in the RMN method, the rotation angle at each iteration must be precomputed on a classical computer under the assumption that the exact value of the ratio $M/N$ is known. Indeed, when $M/N$ is known, there exist variants of Grover's algorithm \cite{hoyer2000arbitrary,long2001zero,bhmt2002amplitude} that achieve success probability exactly equal to $1$, with complexity independent of $\varepsilon$.

Hence, the primary contribution of this work is not an improvement in practical complexity over existing search algorithms, but a new rigorous optimization-based construction and complexity analysis for quantum unstructured search.

\paragraph{Organization.}
In \cref{sec-prelim}, we introduce the preliminaries of the unstructured search problem.
In \cref{sec-rga}, we revisit the Grover-compatible Riemannian gradient ascent method.
In \cref{sec-newton}, we derive the Riemannian Hessian and present the proposed Grover-compatible Riemannian modified Newton method.
Numerical experiments are provided in \cref{sec-experiments}.
We conclude in \cref{sec-discussion}.
The appendix develops the required manifold geometry and rigorously proves the complexity of our algorithm.

\section{Preliminaries}\label{sec-prelim}

We first review the notation for the unstructured search problem and then reformulate it as a maximization problem on the unitary manifold, followed by a discussion of the corresponding optimality conditions. In the subsequent sections, we will propose several algorithms.

\subsection{Problem statement}

Let $\mathcal{H}$ be the $N$-dimensional Hilbert space of an $n$-qubit system with $N=2^n$. We consider an unstructured search problem over the finite set $[N]:= \{0,1,\ldots,N-1\},$ where a binary oracle function $g \colon [N]\to\{0,1\}$ specifies whether an element is marked. The set of marked items is $S:=\{x\in[N]: g(x)=1\}$ with $M:=|S|$ and typically $1\leq M\ll N$.
The goal of the search task is to identify at least one element in $S$.

Let $\{|j\rangle\}_{j=0}^{N-1}$ denote the computational basis of $\mathcal{H}$, where $|j\rangle$ represents the $j$-th standard basis vector.
Define the marked subspace as $\mathcal{T}:= \operatorname{span}_{\mathbb{C}}\{\,|x\rangle: x\in S\,\}\subseteq\mathcal{H}$.
The quantum search algorithm can then be interpreted as the task of preparing a normalized state vector in $\mathcal{T}$.
Let the orthogonal projector onto $\mathcal{T}$ be
\begin{equation}\label{eq-defn-Hf}
    H:= \sum_{x\in S} |x\rangle\langle x|.
\end{equation}
Grover's algorithm begins with the uniform superposition over the entire computational basis, i.e., $|\psi_0\rangle:= \frac{1}{\sqrt{N}}\sum_{x=0}^{N-1} |x\rangle$. We denote the associated rank-one projector by $\psi_0:= |\psi_0\rangle\langle\psi_0|$.
Grover's algorithm uses two types of quantum gates, the \textit{diffusion} operator and the \textit{oracle} operator, defined respectively as
\begin{align}
     D(\alpha) &:=e^{i \alpha \psi_0}=I+\left(e^{i \alpha}-1\right) \psi_0, \label{eq-diffusion} \\
    O_g(\beta) &:=e^{i \beta H}=I+\left(e^{i \beta}-1\right) H,  \label{eq-oracle}
\end{align}
where the parameters $\alpha, \beta \in \mathbb{R}$ denote the angles.
Grover iterations are constructed by alternating these two operators, $G(\alpha_k,\beta_k)=-D(\alpha_k)O_g(\beta_k)$.
After $T$ iterations, the resulting state is $|\psi_{\text{final}}\rangle
= \prod_{k=1}^{T} G(\alpha_k,\beta_k)\,|\psi_0\rangle$, which aims to approximate the ideal target state
\begin{equation}\label{eq-defn-psistar}
    |\psi^\star\rangle
:= \frac{1}{\sqrt{M}}\sum_{x\in S}|x\rangle
    \in\mathcal{T}.
\end{equation}
The query complexity of Grover's algorithm is commonly measured by the number of oracle calls $O_g(\beta)=e^{i\beta H}$.
Detailed circuit implementations of both gates can be found in standard textbooks \cite{kaye2006introduction,nielsen2010quantum}.

In the original version of Grover's algorithm \cite{grover1996fast}, the parameters are chosen as $\alpha_k=\beta_k=\pi$ for all $k$.
However, this choice may lead to the overshooting problem \cite{brassard1997searching}: performing too many Grover iterations may cause the state to rotate past the target rather than toward it.
To mitigate this issue, the $\pi/3$ algorithm \cite{grover2005fixed} fixes the angles at $\alpha_k=\beta_k=\pi/3$, although this modification loses the quadratic speedup.
Later, the fixed-point algorithm \cite{yoder2014fixed} was proposed to avoid overshooting while still retaining the quadratic speedup.

\subsection{Manifold optimization formulation}

In this subsection, we reformulate the quantum search problem as a maximum-eigenstate preparation problem, which can further be cast as a manifold optimization problem. This idea has been proposed in \cite{suzuki2025grover,lai2025grover}. Throughout, we equip the matrix space with the Frobenius inner product $\langle A, B\rangle:=\operatorname{Tr}(A^{\dagger} B)$.

Let us consider the Hamiltonian $H$ defined in \eqref{eq-defn-Hf}. It satisfies $H=H^{\dagger}=H^2$, which implies that its eigenvalues can only take the values 0 and 1. In fact, the spectrum of $H$ consists of eigenvalue 1 on the marked subspace $\mathcal{T}$ with multiplicity $M$, and eigenvalue 0 on its orthogonal complement $\mathcal{T}^{\perp}$. In particular, we have $H\left|\psi^{\star}\right\rangle=\left|\psi^{\star}\right\rangle.$

Let $q_0:=\frac{M}{N} \in(0,1)$ denote the probability of selecting a marked item at random. Applying $H$ to the initial state gives $H\left|\psi_0\right\rangle=\sqrt{q_0}\left|\psi^{\star}\right\rangle.$ Thus, the initial expectation value satisfies $\left\langle\psi_0\right| H\left|\psi_0\right\rangle=\| H\left|\psi_0\right\rangle \|^2=\| \sqrt{q_0}\left|\psi^{\star}\right\rangle \|^2=q_0.$
Indeed, for any state $|\psi\rangle$, the expectation value $\langle\psi| H|\psi\rangle = \sum_{x \in S}|\langle x | \psi\rangle|^2 \in[0, 1]$ represents the probability of observing a marked item when measuring $|\psi\rangle$ in the computational basis.

Based on the above discussion, our goal now is to construct a quantum state whose expectation value with respect to $H$ (i.e., the probability of observing a marked item) equals 1. To achieve this, we start from an easily prepared state $\left|\psi_0\right\rangle$ and aim to design a quantum circuit $U$, whose quantum gates are easily implementable on quantum hardware (e.g., diffusion/oracle operators), such that the resulting state $U\left|\psi_0\right\rangle$ becomes an eigenstate of $H$ corresponding to its largest eigenvalue 1.

Consequently, we consider the following manifold optimization problem:
\begin{equation}\label{eq-problem}
    \max_{U \in \mathrm{U}(N)} f(U),
\end{equation}
where the feasible region is the compact Lie group of $N \times N$ unitary matrices, which forms a Riemannian manifold,
\begin{equation*}
    \mathrm{U}(N):= \left\{ U \in \mathbb{C}^{N \times N} \mid U^\dagger U = I_N \right\},
\end{equation*}
and the cost function $f: \mathrm{U}(N) \rightarrow \mathbb{R}$ is defined by
\begin{equation}\label{eq-cost}
    f(U):= \operatorname{Tr}\!\left(H U \psi_0 U^\dagger \right)
    = \langle \psi_0 | U^\dagger H U | \psi_0 \rangle.
\end{equation}

To solve the above problem on the manifold $\mathrm{U}(N)$, we will employ geometric tools and algorithms from the manifold optimization framework \cite{absil2008optimization,boumal2023introduction}. In the main text that follows, we only briefly describe the geometric tools on $\mathrm{U}(N)$. A detailed introduction to the geometry of the unitary manifold is provided in \cref{app-geo}.

To characterize the directions of infinitesimal motion on the manifold $\mathrm{U}(N)$, we introduce the tangent space.
For any $U \in \mathrm{U}(N)$, the tangent space at $U$ is given by
\begin{equation*}
T_U:= \left\{ \Omega U: \Omega^\dagger = -\Omega \right\}
= \mathfrak{u}(N)U,
\end{equation*}
where
\begin{equation*}
\mathfrak{u}(N):= \left\{ \Omega \in \mathbb{C}^{N \times N}: \Omega^\dagger = -\Omega \right\}
\end{equation*}
denotes the Lie algebra of $\mathrm{U}(N)$, i.e., the real vector space of all $N \times N$ skew-Hermitian matrices.
Intuitively, the tangent space can be viewed as the tangent plane to the manifold at a given point (see \cref{fig:manifold-iteration}).

For an unconstrained optimization problem, it is well known that the optimality condition is that the gradient vanishes. When the problem is restricted to a manifold, this condition becomes the vanishing of the Riemannian gradient (see \cref{app-geo} for more).
Let $\psi_U:= U \psi_0 U^\dagger$ denote the density operator corresponding to the output state of the circuit $U$. The \textit{Riemannian gradient} of $f$ at some $U \in \mathrm{U}(N)$ is given by
\begin{equation*}
\operatorname{grad} f(U) = [H, \psi_U]\, U \in T_U,
\end{equation*}
where $[A, B]:=A B-B A$. For convenience, we denote the corresponding skew-Hermitian part by
\begin{equation*}
\widetilde{\operatorname{grad}} f(U) = [H, \psi_U] \in \mathfrak{u}(N).
\end{equation*}
When discussing the tangent space $T_U=\mathfrak{u}(N) U$, we often focus on $\mathfrak{u}(N)$ itself and omit the right unitary $U$.
\cite[Theorem 3.2]{lai2025grover} has shown that $\widetilde{\operatorname{grad}} f\left(U^{\star}\right)=0$ if and only if $U^{\star} \in \mathrm{U}(N)$ globally minimizes $f(U)$ over $\mathrm{U}(N)$ with the minimum value 0, or globally maximizes it with the maximum value 1.

\begin{remark}
Note that the above global optimality condition, resembling that of a convex function, is quite special. It holds only for Hamiltonians satisfying $H=H^2$. In general cases, it serves merely as a necessary condition for a local extremum.
\end{remark}

\section{Grover-compatible Riemannian gradient ascent method}\label{sec-rga}

In this section, we introduce the Grover-compatible Riemannian gradient ascent (RGA) method for problem \cref{eq-problem}.
The term \textit{gradient ascent} indicates that it is a typical first-order optimization algorithm. The term \textit{Grover-compatible} means that each iteration is implemented using the diffusion and oracle operators (\cref{eq-diffusion,eq-oracle}) from Grover's algorithm, making the method easy to realize on quantum hardware.
This RGA method, originally proposed in the author’s previous work \cite{lai2025grover}, serves as a key foundation for the second-order method, the Grover-compatible Riemannian modified Newton method, developed in this work.

\subsection{From Euclidean iterations to manifold iterations}

Recall that in the classical Euclidean setting, to solve the optimization problem $\max _{x \in \mathbb{R}^n} f(x),$ one typically constructs a sequence $\left\{x_k\right\} \subseteq \mathbb{R}^n$ from an initial iterate $x_0 \in \mathbb{R}^n$ using the update rule
\begin{equation}\label{eq-Euc-update}
    x_{k+1}=x_k+t_k \eta_k,
\end{equation}
where $\eta_k \in \mathbb{R}^n$ is a search direction (e.g., usual gradient $\nabla f(x_k)$) and $t_k>0$ is a step size.
To extend the update scheme \cref{eq-Euc-update} to manifolds, the search direction $\eta_k$ is now chosen in the tangent space, and a mapping called \textit{retraction} is used to map $t_k\eta_k$ back to the manifold.

\begin{figure}[!t]
    \centering
    \includegraphics[width=0.9\linewidth]{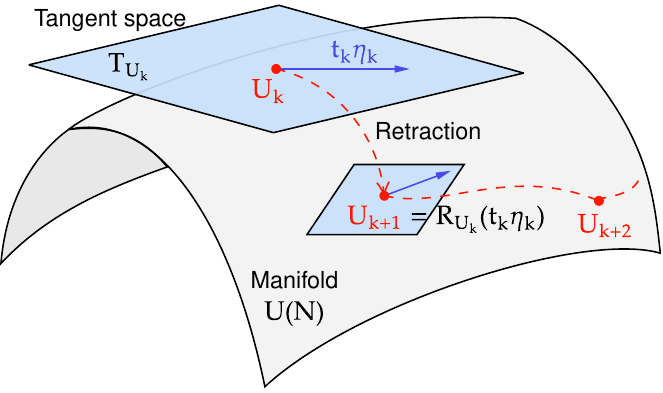}
    \caption{Schematic illustration of a manifold optimization iteration on the unitary manifold $\mathrm{U}(N)$. Starting from the current point $U_k$, a tangent direction $\eta_k \in T_{U_k}$ (e.g., the Riemannian gradient) is chosen in the tangent space. The retraction $\mathrm{R}_{U_k}$ then maps the scaled tangent vector $t_k \eta_k$ back onto the manifold, producing the next iterate $U_{k+1} = \mathrm{R}_{U_k}(t_k \eta_k)$.}
    \label{fig:manifold-iteration}
\end{figure}

Specifically, to solve the manifold optimization problem $\max_{U \in \mathrm{U}(N)} f(U)$, starting from an initial point $U_0 \in \mathrm{U}(N)$, a sequence $\{U_k\} \subseteq \mathrm{U}(N)$ is generated via update rule (see \cref{fig:manifold-iteration} for a schematic illustration):
\begin{equation*}
    U_{k+1} = \mathrm{R}_{U_k}\!\left(t_k \eta_k\right),
\end{equation*}
where $\eta_k \in T_{U_k}$ is a tangent vector at the current point (e.g., Riemannian gradient $\eta_k = \operatorname{grad} f(U_k)$), and $t_k > 0$ is still a step size.
The mapping $\mathrm{R}_U: T_U \rightarrow \mathrm{U}(N)$ is called a
\textit{retraction}, which enables movement along tangent directions while keeping iterates on the manifold. Formally, a retraction is a smooth mapping $\mathrm{R}_U: T_U \rightarrow \mathrm{U}(N)$ such that the induced curve $t \mapsto \gamma(t):= \mathrm{R}_U(t\eta)$ satisfies
\begin{equation}\label{eq-defn-retraction}
\gamma(0) = U, \quad \dot{\gamma}(0) = \eta, \quad \forall \eta \in T_U.
\end{equation}
A natural example of a retraction on $\mathrm{U}(N)$ is the \textit{Riemannian exponential map}, defined by
\begin{equation}\label{eq-Exp}
    \mathrm{R}_U^{\text{exp}}(\eta) = e^{\widetilde{\eta}} U,
    \quad
    \widetilde{\eta}:= \eta U^\dagger \in \mathfrak{u}(N).
\end{equation}
The induced curve then $\gamma(t) = e^{t\widetilde{\eta}} U$ satisfies $\gamma(0) = U$ and $\dot{\gamma}(0) = \widetilde{\eta}U = \eta$, verifying that $\mathrm{R}_U^{\text{exp}}$ is a valid retraction. Other examples are provided in \cref{app-geo}.

At this point, we can consider the simplest first-order optimization algorithm, i.e., the Riemannian gradient ascent method with the retraction $\mathrm{R}=\mathrm{R}^{\text{exp}}$ defined in \cref{eq-Exp}. Its update rule takes the form ($k=0,1,\ldots$)
\begin{equation}\label{eq-Exp-update}
    U_{k+1}
    = \mathrm{R}_{U_k}^{\text{exp}} \!\left(t_k \operatorname{grad} f(U_k)\right)
    = e^{\,t_k [H,\psi_k]} \, U_k,
\end{equation}
where $\psi_k:= U_k \psi_0 U_k^\dagger$ denotes the intermediate quantum state after the $k$-th iteration.
Setting the initial gate to $U_0=I$, the resulting state after $T$ iterations is
\begin{equation*}
    |\psi_T\rangle = e^{t_{T-1}[H, \psi_{T-1}]} \cdots e^{t_1[H, \psi_1]} e^{t_0[H, \psi_0]}|\psi_0\rangle.
\end{equation*}

However, exactly implementing the newly introduced gate $e^{t_k\left[H, \psi_k\right]}$ on quantum hardware is very challenging.
Current works \cite{magann2023randomized,malvetti2024randomized,pervez2025riemannian,lai2026quantum} address this issue by using stochastic subspace projected gradient. In these approaches, the Riemannian gradient $\left[H, \psi_k\right]$ is projected onto a small number of randomly selected Pauli words, which form a basis of a random subspace in the tangent space $\mathfrak{u}(N)$; a Trotter approximation \cite{trotter1959product} is then used as the retraction (see \cref{eq-trt} in \cref{app-geo}). However, such methods can only exploit partial information of the Riemannian gradient at each iteration.
As we will soon see, for the Hamiltonian $H=H^2$ arising in the Grover setting, there exists an exact implementation of the Riemannian gradient, rather than merely an approximate one.

\subsection{Grover-compatible retractions}

The difficulty of implementing the gate $e^{t_k[H,\psi_k]}$ in \cref{eq-Exp-update} suggests that the Riemannian exponential map $\mathrm{R}_{U}=\mathrm{R}^{\text{exp}}_{U}$ in \cref{eq-Exp} may not be a practical choice of retraction.
Note that the convergence guarantees of manifold optimization \cite{absil2008optimization,boumal2023introduction} only require a retraction $\mathrm{R}_{U}$ to satisfy the first-order properties: $\gamma(0)=U$, $\dot{\gamma}(0)=\eta$ in \cref{eq-defn-retraction}; while its specific form is unrestricted. This motivates us to adopt another retraction that is easier to implement on quantum devices.

Throughout the remainder of the paper, let $H=H^{\dagger}=H^2$ be an orthogonal projector on Hilbert space $\mathcal{H}$ and $M= \operatorname{rank}(H)$; let $\left|\psi_0\right\rangle \in \mathcal{H}$ be a unit vector and set $\psi_0 =\left|\psi_0\right\rangle\left\langle\psi_0\right|$. Assume that $\left\langle\psi_0\right| H\left|\psi_0\right\rangle \notin\{0, 1\}$. Define the Frobenius norm by $\|A\|_F^2 = \langle A, A \rangle$.
We begin with an auxiliary lemma.

\begin{lemma}[{\cite[Lemma 3.3]{lai2025grover}}]\label{lem-X0YO}
Let $\psi = |\psi\rangle\langle\psi|$ be any pure state and define $q:= \langle \psi | H | \psi \rangle$.
Define two skew-Hermitian operators by $X:= [H, \psi]$ and $Y:= i[H, X]$.
Then $\|X\|_F = \|Y\|_F = \sqrt{2q(1-q)}$ and $\langle X, Y \rangle = 0$.
\end{lemma}

The following \cref{thm-grad-in-W} shows that if the circuit structure consists only of the diffusion and oracle operators defined in \cref{eq-diffusion,eq-oracle}, then the output states always remain in the Grover plane. More importantly, the corresponding Riemannian gradient also stays within a fixed two-dimensional subspace of the tangent space.
Note that $X_0$ and $Y_0$ defined in \cref{eq-X0YO} are orthogonal, as stated in \cref{lem-X0YO}.

\begin{theorem}[Invariant 2D gradient subspace {\cite[Theorem 3.6]{lai2025grover}}]\label{thm-grad-in-W}
Initialize the unitary by setting $U_0=I$.
 For $k=0,1, \ldots$, consider the update rule:
\begin{equation*}
    V_k:= \text{a finite product of the form } e^{i \theta H} \text{ and } e^{i \theta \psi_0},
\end{equation*}
and define
$U_{k+1}:= V_k U_k$. Moreover, let $\psi_k:=\left|\psi_k\right\rangle\left\langle\psi_k\right|$, where $\left|\psi_k\right\rangle=U_k\left|\psi_0\right\rangle$; equivalently, $\left|\psi_{k+1}\right\rangle=V_k\left|\psi_k\right\rangle$. Define the scalars $q_k:= f(U_k)=\left\langle\psi_k\right|H\left|\psi_k\right\rangle  \in[0, 1],$ and the skew-Hermitian operators
\begin{equation}\label{eq-X0YO}
X_0:=\left[H, \psi_0\right], \quad Y_0:=i\left[H, X_0\right].
\end{equation}
Then, for all $k \geq 0$, the following statements hold:
\begin{enumerate}
    \item The state $|\psi_k\rangle$ remains in the fixed 2-dimensional subspace (called Grover plane)
\begin{equation*}
    |\psi_k\rangle \in \mathcal S:= \operatorname{span}_{\mathbb{C}}\{|\psi_0\rangle, H|\psi_0\rangle\} \subseteq \mathcal{H}.
\end{equation*}
\item The skew-Hermitian part of Riemannian gradient, $[H, \psi_k]$, remains in the fixed 2-dimensional real subspace
\begin{equation*}
    [H, \psi_k] \in \mathcal W:= \operatorname{span}_{\mathbb{R}}\{X_0, Y_0\} \subseteq \mathfrak{u}(N),
\end{equation*}
and $\left\|\left[H, \psi_k\right]\right\|_F=\sqrt{2q_k\left (1-q_k\right)} $.
\end{enumerate}
\end{theorem}

In what follows, whenever there is no ambiguity, the term \textit{gradient} also refers to the skew-Hermitian part of the Riemannian gradient.
\cref{thm-grad-in-W} provides an important insight: if the current gradient lies in the 2-dimensional subspace $\mathcal{W}\subseteq \mathfrak{u}(N)$ (when $U_0 = I$, the initial Riemannian gradient is simply $X_0\in\mathcal{W}$), then it suffices to consider a retraction that is well-defined only on $\mathcal{W}$, rather than on the full $N^2 = 4^n$-dimensional space $\mathfrak{u}(N)$. A typical contrast is the Trotter retraction defined in \cref{eq-trt} of \cref{app-geo}.
First, to ensure that it is a valid retraction, it must satisfy the standard first-order properties in \cref{eq-defn-retraction}. Second, to guarantee that the next Riemannian gradient remains in $\mathcal{W}$, its construction must involve only the form $e^{i\theta H}$ and $e^{i\theta \psi_0}$. In this way, the whole procedure forms a closed loop. This naturally leads to the following definition.

\begin{definition}\label{defn-grover-retraction}
A mapping $\mathrm{R}_U$ is called a \emph{Grover-compatible retraction} if for all $U \in \mathrm{U}(N)$ the following conditions hold:
\begin{enumerate}
    \item $\mathrm{R}_U: \mathcal{W}U \rightarrow \mathrm{U}(N)$ is a well-defined retraction restricted to the domain $\mathcal{W}U \subseteq T_U= \mathfrak{u}(N)U$; that is, the induced curve $\gamma(t) = \mathrm{R}_U(t\eta)$ for $t \ge 0$ satisfies
\begin{equation*}
    \gamma(0) = U, \quad \dot{\gamma}(0) = \eta, \quad \forall \eta \in \mathcal{W}U.
\end{equation*}
    \item $\mathrm{R}_U$ admits the form $\mathrm{R}_U(\eta)
=
V(\eta) U,$ where $V(\eta)$ is a finite product of the form $e^{i\theta H}$ and $e^{i\theta \psi_0}$.
\end{enumerate}
\end{definition}

In fact, retractions satisfying the above definition are easy to construct, and their forms may vary. Here we introduce a $5$-factor retraction in \cref{example-5factor}. According to the number $K$ of factors of the form $e^{i\theta H}$ and $e^{i\theta \psi_0}$ it contains, we refer to it as a $K$-factor retraction. More examples, such as $6$-factor and $8$-factor retractions, can be found in \cite[Proposition~3.10]{lai2025grover}.

\begin{example}[$5$-factor retraction]\label{example-5factor}
For any $U \in \mathrm{U}(N)$ and $\eta \in \mathcal{W}U$, write the skew-Hermitian part $\widetilde{\eta}:= \eta U^\dagger \in \mathcal{W}$ in the form $\widetilde{\eta} = x X_0 + y Y_0$ for unique coefficients $(x,y) \in \mathbb{R}^2$.
Let $A:= \operatorname{atan2}(y,x)$ and $R:= \sqrt{x^2+y^2}$, which correspond to the argument and modulus of the complex number $x+iy$, and define
\begin{equation*}
    a_1:= A + \frac{\pi}{2},
    \;
    a_2:= A - \frac{\pi}{2},
    \;
    b_1:= -\frac{R}{2},
    \;
    b_2:= \frac{R}{2}.
\end{equation*}
The $5$-factor retraction is defined by
\begin{equation}\label{eq-5factor-retraction}
    \mathrm{R}^{(5)}_U(\eta)
:=
    e^{i a_1 H}
    e^{i b_1 \psi_0}
    e^{i(a_2-a_1)H}
    e^{i b_2 \psi_0}
    e^{-i a_2 H}\, U.
\end{equation}
Since $a_2-a_1=-\pi$, the middle $H$-exponential simplifies to $e^{-i\pi H}$. For $t \ge 0$, the associated curve is given by
\begin{equation*}
    \gamma(t)
    =
    \underbrace{
    e^{i a_1 H}
    e^{i t b_1 \psi_0}
    e^{-i\pi H}
    e^{i t b_2 \psi_0}
    e^{-i a_2 H}
    }_{\text{newly added gates } V(t;x,y):=}
    \, U.
\end{equation*}
It can be verified that $\gamma(0)=U$ and $\dot{\gamma}(0)=\eta \in \mathcal{W}U$.
\end{example}

Indeed, the Grover-compatible retractions $\mathrm{R}_U\colon \mathcal{W}U \simeq \mathbb{R}^2 \rightarrow \mathrm{U}(N)$ admit the following general abstract form:
\begin{equation}\label{eq-grover-form}
    \mathrm{R}_U (t\eta)=\underbrace{\left (\prod_{\ell=1}^{K}
    e^{i\, \theta_{\ell}^{ (1)} (t; x, y)\, H}\,
    e^{i\, \theta_{\ell}^{ (2)} (t; x, y)\, \psi_0}\right)}_{\text{newly added gates } V(t;x,y):=} \,U,
\end{equation}
for $\widetilde{\eta}:= \eta U^\dagger =x X_0+y Y_0\in \mathcal{W}$. The order of $H$- and $\psi_0$-exponentials does not matter.
Now, we provide the Riemannian gradient ascent (RGA) method employing Grover-compatible retractions in \cref{alg-grover-ret}.

\begin{algorithm}[t]
\caption{Grover-compatible Riemannian gradient ascent (RGA) method}
\label{alg-grover-ret}

\KwIn{Choose a Grover-compatible retraction $\mathrm{R}$ in \cref{defn-grover-retraction} with formula \cref{eq-grover-form}, initial point $U_0=I$, and tolerance $\varepsilon$.}

 Set $k:= 0$, $q_0:= M/N$, $e_0 := 1-q_0$\;

\While{$e_k > \varepsilon$}{
  Decompose $[H,\psi_k] = x_k X_0 + y_k Y_0$ to find $(x_k,y_k) \in \R^2$\;

  Choose a step size $t_k > 0$\;

  Update $U_{k+1}
:= V(t_k;x_k,y_k)U_k $ by \cref{eq-grover-form}\;

  Update $|\psi_{k+1}\rangle:= V(t_k;x_k,y_k)|\psi_k\rangle$\;

  Set new cost value
  $q_{k+1}:= \langle \psi_{k+1}|H|\psi_{k+1}\rangle$\;

  Set success probability error $e_{k+1}:= 1-q_{k+1}$\;

  Set $k:= k+1$\;
}
\end{algorithm}

\cite[Theorem 4.6]{lai2025grover} provides explicit results on the convergence and complexity of \cref{alg-grover-ret}. Suppose we run \cref{alg-grover-ret} with the 5-factor retraction in \cref{eq-5factor-retraction}, and choose a fixed step size $t_k=1 / L_{\text {Rie}}$, where $L_{\text {Rie}}:=2+N/\sqrt{2 M(N-M)} \in \mathcal{O}(\sqrt{N/M})$. Then, for any $0<\varepsilon \leq q_0$, the iterates satisfy $1-q_T \leq \varepsilon$ within at most $T=\left\lceil 6 L_{\text {Rie}} \log \left(1/\varepsilon\right)\right\rceil$ iterations. In summary, we establish an iteration complexity of $\mathcal{O}(\sqrt{N / M} \log (1 / \varepsilon))$ to reach an $\varepsilon$-global maximizer, which is consistent with Grover's quadratic speedup. The term $\log (1 / \varepsilon)$ indicates that the method converges linearly with respect to the accuracy $\varepsilon$.

\begin{remark}
Note that as long as the initial state $\left|\psi_0\right\rangle$ is the uniform state, \cref{alg-grover-ret} is guaranteed to converge to the desired target state $\left|\psi^{\star}\right\rangle$ defined in \cref{eq-defn-psistar}. Recall that $\left|\psi^{\star}\right\rangle=\left(1 / \sqrt{q_0}\right) H\left|\psi_0\right\rangle$. For any $\left|\psi_k\right\rangle \in \mathcal{S}= \operatorname{span}_{\mathbb{C}}\left\{\left|\psi_0\right\rangle, H\left|\psi_0\right\rangle\right\}$, the cost value satisfies $q_k=\left\langle\psi_k\right| H\left|\psi_k\right\rangle=\left|\left\langle\psi_k | \psi^{\star}\right\rangle\right|^2$. Therefore, when $q_k \rightarrow 1$, the states $\left|\psi_k\right\rangle$ and $\left|\psi^{\star}\right\rangle$ differ only by a global phase.
\end{remark}

\subsection{Classical simulability of \texorpdfstring{\cref{alg-grover-ret}}{Algorithm 1}}

\cref{alg-grover-ret} does not need to be executed directly on a quantum device. As shown in \cref{thm-grad-classical}, the procedure is classically simulable. Intuitively, the entire evolution of \cref{alg-grover-ret} takes place within a two-dimensional subspace. We choose $u:=H\left|\psi_0\right\rangle$ and $v:=(I-H)\left|\psi_0\right\rangle$ as a basis for the Grover plane.
It therefore suffices to represent the state $\left|\psi_k\right\rangle$ and the gradient $\left[H, \psi_k\right]$ in this basis. The detailed proof can be found in \cite[Theorem 3.13]{lai2025grover}. In practice, the gate angles can be computed in advance on a classical computer, after which the resulting circuit is executed on quantum hardware.

\begin{theorem}[Classical simulability of \cref{alg-grover-ret} {\cite[Theorem 3.13]{lai2025grover}}]
\label{thm-grad-classical}
Consider the iterative process in \cref{alg-grover-ret}.
With the initial triplet $(x_0,y_0,q_0):=(1,0,q_0)$ and the $2\times 2$ matrix $\Psi_0 = \begin{bmatrix} q_0 & 1-q_0\\ q_0 & 1-q_0 \end{bmatrix}$, there exists an explicit, classically computable process
\begin{equation*}
    F_{\mathrm{RGA}}: (x_k, y_k, q_k; t_k) \mapsto (x_{k+1}, y_{k+1}, q_{k+1}),
\end{equation*}
described as follows. Initialize $\alpha_0=\beta_0=1$, and $z_0=\alpha_0\bar{\beta}_0.$ For each $k = 0, 1, \ldots,$

\begin{enumerate}
    \item For the additional update gates $V (t_k; x_k, y_k)$ in  \cref{eq-grover-form}, define the corresponding $2 \times 2$ matrix
\begin{equation*}
    M_k:=\prod_{\ell=1}^{K}
    E_H\! \big (\theta_{\ell}^{ (1)} (t_k; x_k, y_k)\big)\,
    E_{\psi_0}\! \big (\theta_{\ell}^{ (2)} (t_k; x_k, y_k)\big),
\end{equation*}
where $E_H(\theta)=
\begin{bmatrix}
e^{i \theta} & 0 \\
0 & 1
\end{bmatrix}$, $E_{\psi_0}(\theta) = I_2 + (e^{i\theta} - 1)\Psi_0$.
\item Update $\begin{pmatrix}\alpha_{k+1}\\ \beta_{k+1}\end{pmatrix}=M_k
\begin{pmatrix}\alpha_k\\ \beta_k\end{pmatrix}$, and let $q_{k+1}=q_0|\alpha_{k+1}|^2,$ $z_{k+1}=\alpha_{k+1} \bar{\beta}_{k+1},$ $x_{k+1}=\Re (z_{k+1}),$ $y_{k+1}=\Im (z_{k+1})$.
\item Set $k\leftarrow k+1$ and return to step 1.
\end{enumerate}
\end{theorem}
As we will see later, the Riemannian Newton method proposed in this work is also classically simulable and can be obtained with modifications based on \cref{thm-grad-classical}.

\section{Grover-compatible Riemannian modified Newton method}\label{sec-newton}

In the classical Euclidean setting, consider the optimization problem $\max _{x \in \mathbb{R}^n} f(x)$.
A prototypical second-order algorithm is the Newton method \cite{dennis1996numerical,kelley2003solving,nocedal2006numerical}, which updates via the Hessian matrix:
\begin{equation*}
x_{k+1}=x_k+d_k, \quad d_k:=-\left[\nabla^2 f\left(x_k\right)\right]^{-1} \nabla f\left(x_k\right),
\end{equation*}
where $d_k$ denotes the Newton direction. Under standard regularity conditions, the method exhibits quadratic convergence in a neighborhood of the solution. The Newton method has also been extended to optimization on manifolds \cite{adler2002newton,absil2008optimization,ferreira2012local}, while preserving the same quadratic convergence rate with respect to the accuracy $\varepsilon$.

In this section, building on the first-order framework developed in the previous section, we will replace the Riemannian gradient direction with the Riemannian Newton direction, thereby obtaining a Grover-compatible Riemannian Newton method. As proved in \cref{app-complexity}, this modification improves the complexity from $\mathcal{O}(\sqrt{N / M} \log (1 / \varepsilon))$ to $\mathcal{O}(\sqrt{N / M} +\log \log (1 / \varepsilon))$.

\subsection{Riemannian Hessian}

For $f: \mathbb{R}^n \rightarrow \mathbb{R}$, the Hessian matrix $\nabla^2 f(x)$ at a point $x$ captures the second-order information of $f$. It is an $n \times n$ symmetric matrix and can be viewed as a self-adjoint linear operator on $\mathbb{R}^n$. By analogy, when $f$ is defined on a manifold, the Riemannian Hessian (operator) can be regarded as a self-adjoint linear operator acting on the tangent space. For the cost function defined in \cref{eq-cost}, we have the following lemma.

\begin{lemma}[{\cite[Proposition 1]{lai2026quantum}}]\label{lem-hessian}
For the cost function $f: \mathrm{U}(N) \rightarrow \mathbb{R}$ defined by $f(U)=\operatorname{Tr}(H U \psi_0 U^\dagger)$, the Riemannian Hessian of $f$ at $U \in \mathrm{U}(N)$ is the self-adjoint linear map $\operatorname{Hess} f(U): T_U \rightarrow T_U$,
\begin{equation}\label{eq-hess-1}
    \operatorname{Hess} f(U)[\Omega U]
    =
    \frac{1}{2}\big([H,[\Omega,\psi_U]] + [[H,\Omega],\psi_U]\big)U,
\end{equation}
where $\psi_U=U \psi_0 U^{\dagger}$ and $T_U=\mathfrak{u}(N)  U$. Identifying $T_U \simeq \mathfrak{u}(N)$ yields $\widetilde{\operatorname{Hess}} f(U): \mathfrak{u}(N) \rightarrow \mathfrak{u}(N),$
\begin{equation*}
    \widetilde{\operatorname{Hess}} f(U)[\Omega]
    =
    \frac{1}{2}\big([H,[\Omega,\psi_U]] + [[H,\Omega],\psi_U]\big),
\end{equation*}
which is self-adjoint on $\mathfrak{u}(N)$.
\end{lemma}

The self-adjointness means that $\langle\Xi, \widetilde{\operatorname{Hess}} f(U)[\Omega]\rangle=\langle\widetilde{\operatorname{Hess}} f(U)[\Xi], \Omega\rangle$ for any $\Omega, \Xi \in \mathfrak{u}(N)$.
A derivation of this Hessian expression can be found in \cite[Proposition 1]{lai2026quantum}. An equivalent definition is based on the second-order Taylor expansion on the manifold:
\begin{align*}
& f\left(\mathrm{R}_U(t \eta)\right) \\
& =f(U)+t\langle\operatorname{grad} f(U), \eta\rangle+\frac{t^2}{2}\langle\eta, \operatorname{Hess} f(U)[\eta]\rangle+ \mathcal{O}(t^3),
\end{align*}
where $\mathrm{R}_U$ is a retraction and $\eta \in T_U$.
In \cref{app-geo}, we verify that \cref{eq-hess-1} ensures the above expansion holds.

In the classical Euclidean setting, the Newton direction $d_k$ is defined as the solution to the Newton equation:
\begin{equation*}
    \nabla^2 f\left (x_k\right) d_k=-\nabla f\left (x_k\right).
\end{equation*}
By analogy, the Riemannian Newton direction $\Omega_k^{\mathrm{N}} U_k \in T_{U_k}$ is defined as the solution to the Riemannian Newton equation: $\operatorname{Hess} f (U_k)[\Omega_k^{\text{N}} U_k] = - \grad f (U_k)$, i.e.,
\begin{equation}\label{eq-Newton2}
  \widetilde{\operatorname{Hess}} f (U_k)[\Omega_k^{\text{N}}] =- \widetilde{\grad} f(U_k).
\end{equation}
Solving these equations is generally challenging. In \cite{lai2026quantum}, Pauli words are used as a basis of $\mathfrak{u}(N)$ to convert the above operator equation into a matrix equation. However, this approach is computationally expensive.

In contrast, in the Grover setting, where the Hamiltonian satisfies $H=H^2$, a remarkable relationship arises: the Riemannian gradient is always an eigenvector of the Riemannian Hessian. Consequently, the Newton equation admits an explicit solution. This result is formalized in the following theorem.

\begin{theorem}\label{thm-hessian}
Let \(H=H^\dagger=H^2\) be an orthogonal projector, and let \(\psi=|\psi\rangle\langle\psi|\) be an arbitrary pure state.
Define \(q:=\operatorname{Tr}(H\psi)\in[0,1]\) and \(g:=[H,\psi]\).
Consider the map \(L(\Omega):=\frac12\bigl([H,[\Omega,\psi]]+[[H,\Omega],\psi]\bigr)\).
Then \(L(g)=(1-2q)g\).
\end{theorem}

\begin{proof}
Substituting \(g=[H,\psi]\) into the definition of \(L\), we get \(L(g)=\frac12\bigl([H,[g,\psi]]+[[H,g],\psi]\bigr)\).
By the Jacobi identity, \([H,[g,\psi]]=[[H,g],\psi]+[g,[H,\psi]]=[[H,g],\psi]\), since \([H,\psi]=g\) and \([g,g]=0\).
Hence \(L(g)=[[H,g],\psi]\).
Using \(H^2=H\), we compute \([H,g]=[H,[H,\psi]]=H\psi-2H\psi H+\psi H\).
Therefore \(L(g)=[H\psi-2H\psi H+\psi H,\psi]\).
Using \(\psi^2=\psi\) and \(\psi H\psi=q\psi\), we have \([H\psi,\psi]=H\psi-q\psi\), \([\psi H,\psi]=q\psi-\psi H\), and \([H\psi H,\psi]=qH\psi-q\psi H\).
Combining them gives
\begin{align*}
L(g)&=(H\psi-q\psi)-2(qH\psi-q\psi H)+(q\psi-\psi H)\\
&=(1-2q)(H\psi-\psi H)=(1-2q)g.
\end{align*}
Thus \(L(g)=(1-2q)g\), which proves the claim.
\end{proof}

Note that the above result holds only under the condition $H=H^2$, and does not hold in general.

\subsection{Grover-compatible Riemannian modified Newton method}

In this subsection, we are now ready to present a Riemannian Newton method for problem \cref{eq-problem}, which is both Grover-compatible and classically simulable, similar to \cref{alg-grover-ret}.
Let $q_k=f\left (U_k\right)$, and define
\begin{equation*}
\mathrm{g}_k:=\widetilde{\operatorname{grad}} f\left(U_k\right) \neq 0, \quad \mathrm{H}_k:=\widetilde{\operatorname{Hess}} f\left(U_k\right).
\end{equation*}
Then \cref{thm-hessian} shows that
\begin{equation*}
    \mathrm{H}_k\left[\mathrm{g}_k\right]=\lambda_k \mathrm{g}_k, \quad \lambda_k:=1-2 q_k.
\end{equation*}
Since $q_k \in[0, 1]$, the eigenvalue $\lambda_k$ lies in the interval $[-1, 1]$. The original Newton equation $\mathrm{H}_k\left[\Omega_k^{\mathrm{N}}\right]=-\mathrm{g}_k$ in \cref{eq-Newton2} therefore admits the explicit solution
\begin{equation}\label{eq-solution-1}
    \Omega_k^{\mathrm{N}}=\frac{1}{-\lambda_k} \mathrm{g}_k \in \mathcal{W},
\end{equation}
whenever $q_k \neq \frac{1}{2}$. (Recall that $\mathcal{W}:=\operatorname{span}_{\mathbb{R}}\left\{X_0, Y_0\right\} \subseteq \mathfrak{u}(N)$ as defined in \cref{thm-grad-in-W}.) The corresponding standard Newton update is then given by $U_{k+1}=\mathrm{R}_{U_k}(\Omega_k^{\mathrm{N}})$ with unit step size $t_k=1.$

\begin{remark}
Since the Riemannian gradient is always an eigenvector of the Riemannian Hessian, the Newton direction is collinear with the gradient. Then, the Newton method can be viewed as a scaled gradient method.
\end{remark}

However, the standard Newton method converges only when the initial point $U_0$ is sufficiently close to the solution. To improve numerical stability, we therefore employ a modified Newton method \cite{nocedal2006numerical}. Specifically, we introduce a parameter $\mu_k \geq 0$ to modify the Hessian operator and solve the modified Newton equation
\begin{equation*}
    \left (\mathrm{H}_k-\mu_k \mathrm{I}\right)\left[\Omega_k^{\mathrm{N}}\right]=-\mathrm{g}_k,
\end{equation*}
where $\mathrm{I}$ denotes the identity operator on the tangent space $\mathfrak{u}(N)$.
By \cref{thm-hessian}, this equation also admits a solution, namely the modified Newton direction
\begin{equation*}
    \Omega_k^{\mathrm{N}}=\gamma_k \mathrm{g}_k \in \mathcal{W}, \quad \gamma_k:=\frac{1}{\mu_k-\lambda_k},
\end{equation*}
whenever $\mu_k\neq\lambda_k $.
If we set $\mu_k>\lambda_k$, then
\begin{equation*}
\left\langle\Omega_k^{\mathrm{N}}, \mathrm{g}_k\right\rangle=\gamma_k\left\|\mathrm{g}_k\right\|^2=\gamma_k \, 2 q_k\left(1-q_k\right)>0.
\end{equation*}
Hence, the modified Newton direction is an ascent direction, and there always exists a step size $t_k>0$ along this direction that increases the cost value. The parameter $\mu_k$ is determined according to the following two cases:
\begin{enumerate}
    \item If $q_k>0.5$, then $\lambda_k<0$. In this case, the Hessian operator is already negative along the gradient direction. We therefore set $\mu_k=0$, which yields $\gamma_k=\frac{1}{-\lambda_k}$ like \cref{eq-solution-1}. Hence, no correction is needed.
    \item  If $q_k \leq 0.5$, then $\lambda_k \geq 0$. In this case, the Hessian operator is positive along the gradient direction, and the unmodified Newton step would point toward a descent direction. To avoid this, we set $\mu_k=\lambda_k+\delta$ with $\delta>0$, which yields $\gamma_k= \frac{1}{\delta}$.
\end{enumerate}
Combining the two cases above, we set the scaling factor
\begin{equation*}
    \gamma_k=\frac{1}{\max \left(\delta, 2 q_k-1\right)}
\end{equation*}
where $\delta>0$ is a small constant, typically chosen in the range $10^{-3}\sim 10^{-6}$.

After ensuring that the Newton direction $\Omega_k^{\mathrm{N}}$ is an ascent direction, we further employ an Armijo backtracking line search; see \cite[Algorithm 3.2]{nocedal2006numerical}. Specifically, we seek the smallest integer $m \geq 0$ such that $t_k=\rho^m$ satisfies
\begin{equation*}
    f\left(\mathrm{R}_{U_k}(t_k \Omega_k^{\mathrm{N}})\right) \geq f\left(U_k\right)+c \, t_k\left\langle\Omega_k^{\mathrm{N}}, \mathrm{g}_k\right\rangle,
\end{equation*}
where $\rho=0.5$ and $c=10^{-4}$. Note that $\left\langle\Omega_k^{\mathrm{N}}, \mathrm{g}_k\right\rangle>0$, which guarantees that the above procedure terminates after finitely many iterations; see \cite[Lemma 3.1]{nocedal2006numerical}. Consequently, the cost value increases monotonically. Finally, a complete description of the proposed algorithm is given in \cref{alg-newton-armijo}.

This modified Newton method with line search ensures convergence from any initial point (i.e., it enjoys the global convergence property of first-order gradient methods), while achieving quadratic convergence once the iterate is sufficiently close to the solution. Intuitively, during the early iterations, \cref{alg-newton-armijo} behaves similarly to \cref{alg-grover-ret}. When the iterate is sufficiently close to the solution (so that $q_k>0.5$ and no direction modification is required), the Armijo condition is always satisfied with $t_k=1$; see \cite[Theorem 3.6 \& 3.8]{nocedal2006numerical}. Consequently, \cref{alg-newton-armijo} reduces to the standard Newton method.

For the rigorous convergence and complexity analysis, see \cref{app-complexity}, where \cref{thm:global-rmn-complexity-final} gives the final complexity bound $\mathcal{O}(\sqrt{N / M} +\log \log (1 / \varepsilon))$.

\begin{remark}[Iteration complexity and query complexity]
The complexity stated in \cref{thm:global-rmn-complexity-final} is an iteration complexity, as is standard in optimization. In the quantum-search setting, however, complexity is often measured by the number of oracle calls \(O_g(\beta)=e^{i\beta H}\), namely the query complexity. For the \(5\)-factor Grover-compatible retraction used in this paper, each RMN update contributes only a constant number of oracle calls; more precisely, after merging adjacent endpoint oracle factors between consecutive updates, \(T\) iterations require \(2T+1\) oracle calls. Hence the iteration complexity and the query complexity differ only by a constant factor determined by the chosen retraction.
\end{remark}

\begin{algorithm}[t]
\caption{Grover-compatible Riemannian modified Newton (RMN) method}
\label{alg-newton-armijo}

\KwIn{Choose a Grover-compatible retraction $\mathrm{R}$ in \cref{defn-grover-retraction} with formula \cref{eq-grover-form}, initial point $U_0=I$, damping $\delta = 10^{-3}$, Armijo parameters $c = 10^{-4}$, $\rho = \tfrac{1}{2}$, and tolerance $\varepsilon$.}

Set $k:= 0$, $q_0:= M/N$, $e_0 := 1-q_0$ and initial squared gradient norm $G_0:= 2q_0(1-q_0)$\;

\While{$e_k > \varepsilon$}{

    Decompose $[H,\psi_k] = x_k X_0 + y_k Y_0$ to find $(x_k,y_k) \in \mathbb{R}^2$\;

    Set modified Newton scaling $\gamma_k:= \frac{1}{\max(\delta,\, 2q_k - 1)}$\;

    Set step size $t:= 1$\;
    \Repeat{$q_{\mathrm{trial}} \ge q_k + c\, t\, \gamma_k\, G_k$}{
        Set $|\psi_{\mathrm{trial}}\rangle:= V(t\gamma_k; x_k, y_k)|\psi_k\rangle$ by \cref{eq-grover-form}\;
        Compute $q_{\mathrm{trial}}:= \langle \psi_{\mathrm{trial}} | H | \psi_{\mathrm{trial}} \rangle$\;
        \If{$q_{\mathrm{trial}} < q_k + c\, t\, \gamma_k\, G_k$}{
            Set $t:= \rho t$\;
        }
    }

    Update $U_{k+1}:= V(t\gamma_k; x_k, y_k)U_k$\;

    Update $|\psi_{k+1}\rangle:=|\psi_{\mathrm{trial}}\rangle$\;

    Set new cost value $q_{k+1}:= q_{\mathrm{trial}}$\;

    Set success probability error $e_{k+1}:= 1-q_{k+1}$\;

    Set $G_{k+1}:= 2q_{k+1}(1-q_{k+1})$\;

    Set $k:= k+1$\;
}
\end{algorithm}

\subsection{Classical simulability of \texorpdfstring{\cref{alg-newton-armijo}}{Algorithm 2}}

Compared with the gradient ascent method in \cref{alg-grover-ret}, the modified Newton method in \cref{alg-newton-armijo} differs only by the introduction of a scaling factor $\gamma_k$ and a line search procedure. Building on the classical simulability of \cref{alg-grover-ret} established in \cref{thm-grad-classical}, we can directly derive a classically simulable version of \cref{alg-newton-armijo}, as stated in the following theorem.
All operations reduce to computations with $2 \times 2$ matrices.
It should be emphasized that the line search procedure in \cref{alg-newton-armijo} is carried out through the following classical computation, rather than through repeated quantum measurements and sampling; therefore, it does not incur any additional quantum overhead.

\begin{theorem}
[Classical simulability of \cref{alg-newton-armijo}]
\label{thm-newton-classical}
Consider the iterative process in \cref{alg-newton-armijo}.
With the initial triplet $(x_0,y_0,q_0):=(1,0,q_0)$ and the $2\times 2$ matrix $\Psi_0 = \begin{bmatrix} q_0 & 1-q_0\\ q_0 & 1-q_0 \end{bmatrix}$, there exists an explicit, classically computable process
\begin{equation*}
    F_{\mathrm{RMN}}:(x_k,y_k,q_k)\mapsto (x_{k+1},y_{k+1},q_{k+1}),
\end{equation*}
described as follows. Initialize $G_0 = 2q_0(1-q_0),$ $\alpha_0=\beta_0=1$, and $z_0=\alpha_0\bar{\beta}_0.$ For each $k=0,1,\dots$,

\begin{enumerate}
\item Set modified scaling $\gamma_k:=\tfrac{1}{\max(\delta,\,2q_k-1)}.$
\item For the trial step $t$, and additional update gates $V (t \gamma_k; x_k, y_k)$ in  \cref{eq-grover-form}, define the corresponding $2 \times 2$ matrix
\begin{equation*}
    M_k(t):=
    \prod_{\ell=1}^{K}
    E_H\!\big(\theta_{\ell}^{(1)}(t\gamma_k;x_k,y_k)\big)\,
    E_{\psi_0}\!\big(\theta_{\ell}^{(2)}(t\gamma_k;x_k,y_k)\big),
\end{equation*}
where $E_H(\theta)=
\begin{bmatrix}
e^{i \theta} & 0 \\
0 & 1
\end{bmatrix}$, $E_{\psi_0}(\theta) = I_2 + (e^{i\theta} - 1)\Psi_0$.
\item For the trial step $t$, define $\begin{pmatrix}
\alpha_k^{\mathrm{temp}}(t)\\[2mm]
\beta_k^{\mathrm{temp}}(t)
\end{pmatrix}
:=
M_k(t)
\begin{pmatrix}
\alpha_k\\[2mm]
\beta_k
\end{pmatrix},$ and $q_k^{\mathrm{temp}}(t):=q_0\big|\alpha_k^{\mathrm{temp}}(t)\big|^2.$
\item Let \(m_k \geq 0 \) be the smallest nonnegative integer such that the Armijo condition
\begin{equation*}
    q_k^{\mathrm{temp}}(\rho^{m_k})
    \;\ge\;
    q_k + c\,\rho^{m_k}\gamma_k\,G_k
\end{equation*}
holds, and set $t_k:=\rho^{m_k}.$
\item Update $\begin{pmatrix}
\alpha_{k+1} \\[2mm]
\beta_{k+1}
\end{pmatrix}
=
M_k(t_k)
\begin{pmatrix}
\alpha_k \\[2mm]
\beta_k
\end{pmatrix}$, and let
$q_{k+1} = q_0 |\alpha_{k+1}|^2$,
$G_{k+1} = 2q_{k+1}(1-q_{k+1})$,
$z_{k+1} = \alpha_{k+1}\bar{\beta}_{k+1}$,
$x_{k+1} = \Re(z_{k+1})$,
$y_{k+1} = \Im(z_{k+1})$.
\item Set $k\leftarrow k+1$ and return to step 1.
\end{enumerate}
\end{theorem}

\section{Numerical experiments}\label{sec-experiments}

In this section, we compare the following methods via numerical simulations:
\begin{itemize}
\item the Grover-compatible Riemannian gradient ascent (\textbf{RGA}) method (\cref{alg-grover-ret}) \cite{lai2025grover}, and
\item the proposed Grover-compatible Riemannian modified Newton (\textbf{RMN}) method (\cref{alg-newton-armijo}).
\end{itemize}
All experiments are implemented in the NumPy package and conducted on a MacBook Pro (2021) with an Apple M1 Pro processor and 16 GB of memory.

Consider an unstructured search problem of size $N=2^n$ for an $n$-qubit system, with a single marked item ($M=1$).
We always use 5-factor retraction defined in \cref{example-5factor}. The initial state $|\psi_0\rangle$ is the uniform superposition.
The progress of the algorithms is measured by the cost value $q_k=f\left(U_k\right)=\operatorname{Tr}\left(H \psi_k\right),$
which represents the success probability of identifying the marked item and should converge to 1 under RGA and RMN.
Given $\varepsilon>0$ (typically $10^{-10}$), the algorithms terminate when success probability error $e_k:=1-q_k<\varepsilon$. Unless otherwise stated, we use the above as the default setting.

\paragraph{Correctness of classical simulability.}

We first examine whether the classical procedures described in \cref{thm-grad-classical,thm-newton-classical} can faithfully simulate the quantum implementations of \cref{alg-grover-ret,alg-newton-armijo}.
To this end, we set $n=4$ qubits.
At each iteration, we evaluate three quantities: the cost value $q_k$, and the coefficients $x_k, y_k$. Recall that $x_k$ and $y_k$ are the expansion coefficients of the Riemannian gradient $\left[H, \psi_k\right]$ on the two-dimensional subspace $\mathcal{W}$ with respect to the basis $\left\{X_0, Y_0\right\}$.

\cref{fig:task1_correctness_classical_vs_matrix} reports the absolute errors of these three quantities for RGA ((a)--(c)) and RMN ((d)--(f)), with and without classical simulation.
Here, ``without classical simulation'' refers to a direct NumPy implementation using explicit $N \times N$ full matrix operations, following \cref{alg-grover-ret,alg-newton-armijo} in a brute-force manner.
The results show that the classical simulation accurately reproduces the true iterations up to machine precision $(10^{-16})$.
Therefore, all subsequent experiments are conducted under classical simulation.

\begin{figure*}
    \centering
    \includegraphics[width=0.9\linewidth]{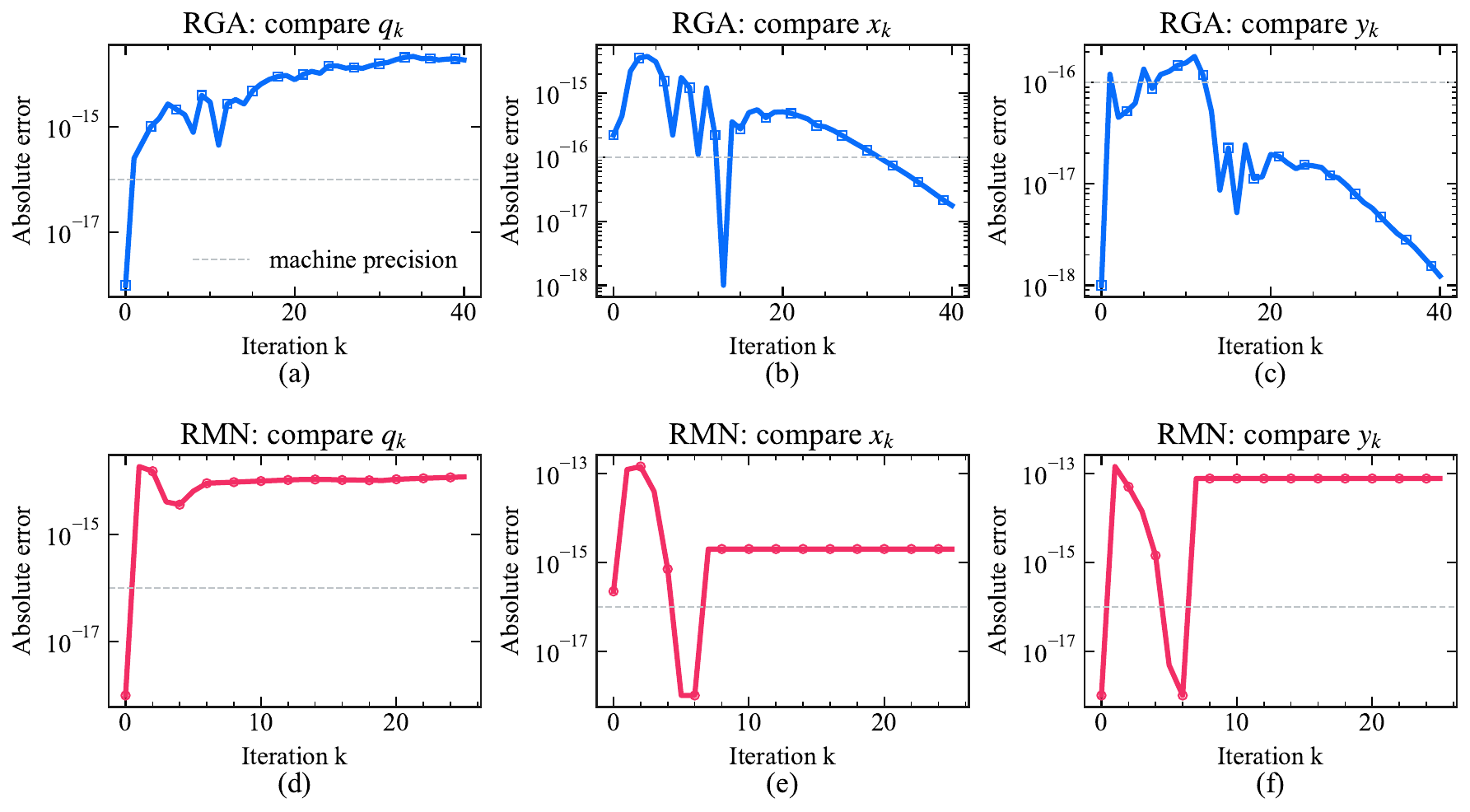}
    \caption{Absolute errors of the cost value $q_k$ and expansion coefficients $x_k, y_k$ between the classical simulation and the explicit full matrix implementation. (a)--(c) display the results for RGA, and (d)--(f) for RMN. All errors remain around machine precision ($10^{-16}$), verifying that the classical procedures accurately simulate the algorithms.}
    \label{fig:task1_correctness_classical_vs_matrix}
\end{figure*}

\paragraph{Comparison of convergence rates.}

As a first-order optimization method, RGA has been shown to converge linearly with respect to the success probability error $e_k:=1-q_k$ \cite{lai2025grover}. Then there exist an integer \(k_1\) and a contraction factor \(\rho\in(0,1)\) such that \(e_{k+1}\le \rho e_k\) for all \(k\ge k_1\).
In contrast, the proposed RMN method is second-order and achieves quadratic convergence: there exist an integer \(k_2\) and a constant \(C\) such that \(e_{k+1}\le C e_k^2\) for all \(k\ge k_2\). For more, see \cref{app-complexity}.

It is well known that quadratic convergence is significantly faster than linear convergence.
For each case, we run both RGA (with fixed step size 0.5) and RMN.
\cref{fig:task2_multi_qubit_comparison} shows the results on a log scale, including the gradient norm ((a)--(c)) and the success probability error ((d)--(f)).
The results show that RGA decreases linearly with respect to the iteration count. In contrast, RMN exhibits quadratic convergence and reaches a high-accuracy solution within only a few iterations.
This behavior is evident in both the gradient norm and the success probability error. For example, in (d), RMN reduces the error from $10^{-2}$ to $10^{-4}$, and then to $10^{-8}$ in just two iterations.

\begin{remark}[Gradient norm versus success probability error]
The two quantities shown in \cref{fig:task2_multi_qubit_comparison} are directly related.
Let \(\mathrm g_k:=[H,\psi_k]\). Since \(\|\mathrm g_k\|_F^2=2q_k(1-q_k)=2q_ke_k\) and \(q_k\to1\) near convergence, we have
\[
        \|\mathrm g_k\|_F^2\sim 2e_k,
        \qquad
        \|\mathrm g_k\|_F\sim \sqrt{2e_k}.
\]
Thus the gradient norm \(\|\mathrm g_k\|_F\) and the success probability error \(e_k\) differ by a square-root scaling. This changes the vertical scale in the plots, but it does not change the convergence type; this is most visible by comparing panels (a) and (d) in \cref{fig:task2_multi_qubit_comparison}. In particular, linear convergence of \(e_k\) still gives linear convergence of \(\|\mathrm g_k\|_F\), while quadratic convergence of \(e_k\) still gives quadratic convergence of \(\|\mathrm g_k\|_F\). Similarly, using the stopping criterion \(\|\mathrm g_k\|_F\le\varepsilon\) only changes the corresponding success probability tolerance to order \(\varepsilon^2\), which does not affect the double-logarithmic dependence because $\log\log\frac1{\varepsilon^2} = \mathcal O\!\left(\log\log\frac1\varepsilon\right).$ For a uniform presentation, throughout this paper we state the convergence and complexity results in terms of the success probability error \(e_k=1-q_k\), rather than the gradient norm.
\end{remark}

\begin{figure*}
    \centering
    \includegraphics[width=0.9\linewidth]{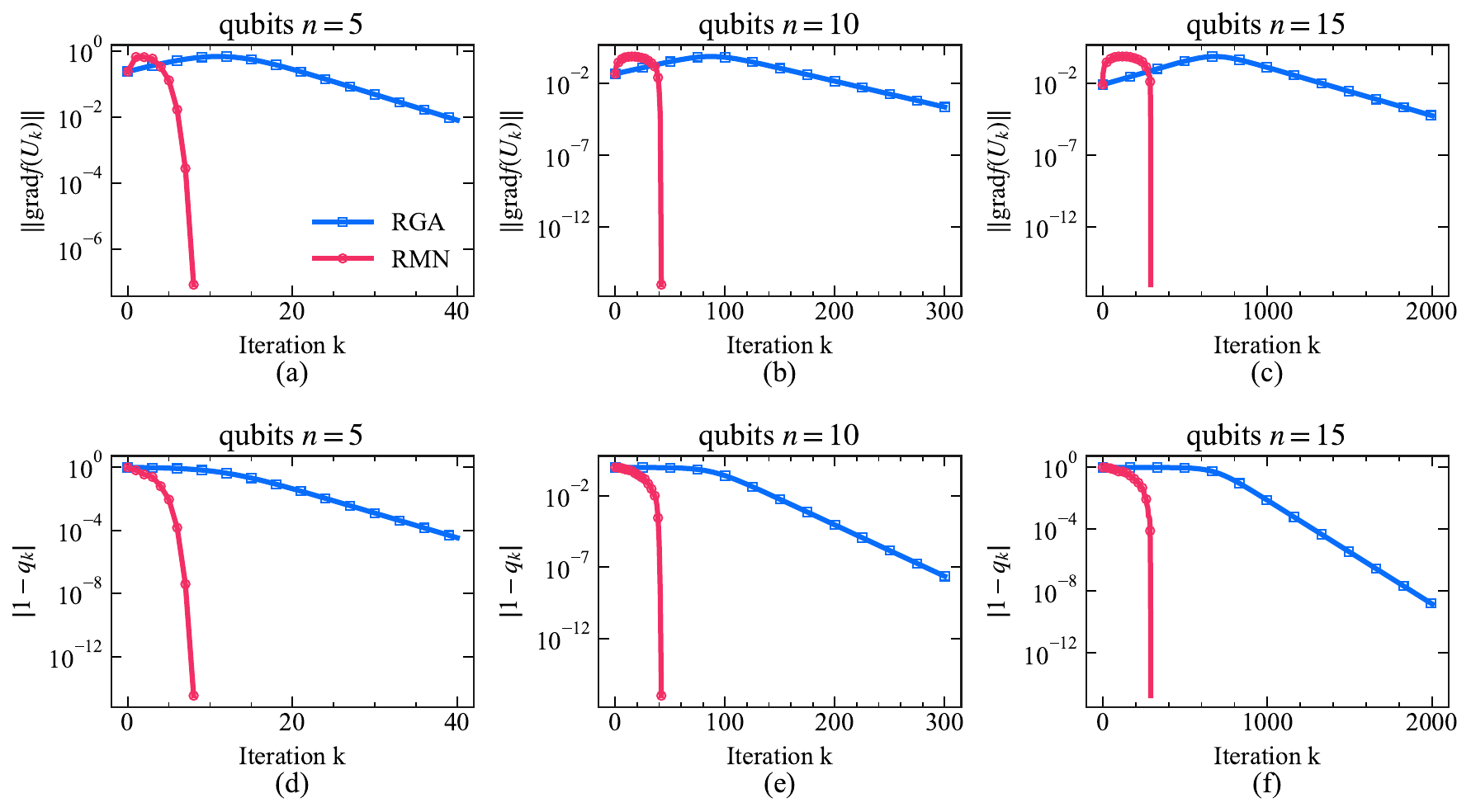}
    \caption{Convergence comparison between the RGA and RMN methods for problem sizes of $n=5$, $10$, and $15$ qubits. (a)--(c) illustrate the gradient norm, and (d)--(f) show the success probability error. The results demonstrate the linear convergence of RGA and the significantly faster, quadratic convergence of RMN.}
    \label{fig:task2_multi_qubit_comparison}
\end{figure*}

\paragraph{Scaling with problem size.}
RGA is shown in \cite{lai2025grover} to have complexity $\mathcal{O}(\sqrt{N / M} \log (1 / \varepsilon))$. Fixing the accuracy $\varepsilon$ and setting $M=1$, this reduces to the well-known quadratic speedup $\mathcal{O}(\sqrt{N})$, i.e., the iteration count grows linearly with $\sqrt{N}$.
The proposed RMN is a second-order optimization method, whose improvement targets the dependence on the accuracy $\varepsilon$ rather than on the problem size $N$. It is therefore expected to retain the same scaling $\mathcal{O}(\sqrt{N})$ with respect to the problem size. However, due to its quadratic convergence, the overall complexity improves to $\mathcal{O}(\sqrt{N}+ \log \log (1 / \varepsilon)).$

Next, we investigate the dependence of RMN on the problem size $N$. We fix $\varepsilon=10^{-6}$ and vary the number of qubits from $n=2$ to $n=28$, corresponding to $\sqrt{N}$ ranging from 2 to $\sqrt{2^{28}}=16384$.
\cref{fig:task3_single_top_complexity} shows a clear linear relationship: the number of iterations of RMN to reach the prescribed accuracy scales proportionally to $\sqrt{N}$.

\begin{figure}
    \centering
    \includegraphics[width=0.8\linewidth]{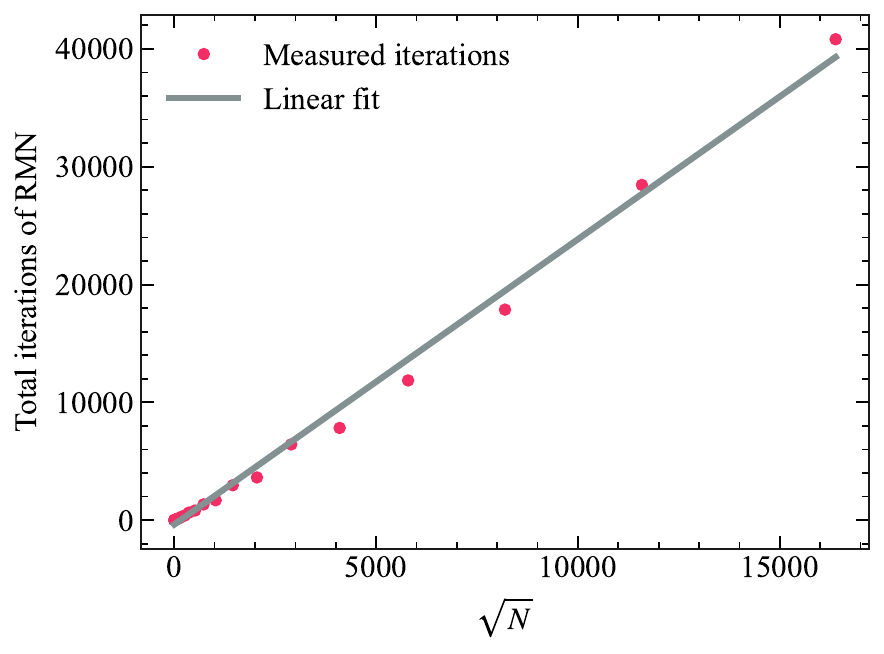}
    \caption{Iteration complexity of RMN versus the square root of the problem size. For a fixed tolerance $\varepsilon$ and varying qubit counts $n \in [2, 28]$ ($N=2^n$), the required iterations for RMN scale linearly with $\sqrt{N}$. This confirms that the second-order RMN method successfully retains the $\mathcal{O}(\sqrt{N})$ quantum speedup.}
    \label{fig:task3_single_top_complexity}
\end{figure}

\paragraph{Multiple marked items.}

We next test whether the same convergence behavior persists beyond the single marked item case. We fix \(n=10\) qubits and set \(M=1,2,4,16,32,64\). \cref{fig:task4_multi_marked_items} reports the success probability error \(e_k\) for these six ratios \(M/N\). Across all panels, RGA still exhibits linear convergence, whereas RMN retains quadratic convergence. In the underlying runs, the RMN iteration count decreases from \(42\) for \(M=1\) to \(8\) for \(M=32,64\), compared with a decrease from \(471\) to \(79\) iterations for RGA. Thus, for the tested \(M>1\) cases, the implementation behaves qualitatively the same as in the single marked item case, and the total iteration count decreases noticeably as \(M\) increases.

\begin{figure*}
    \centering
    \includegraphics[width=0.9\linewidth]{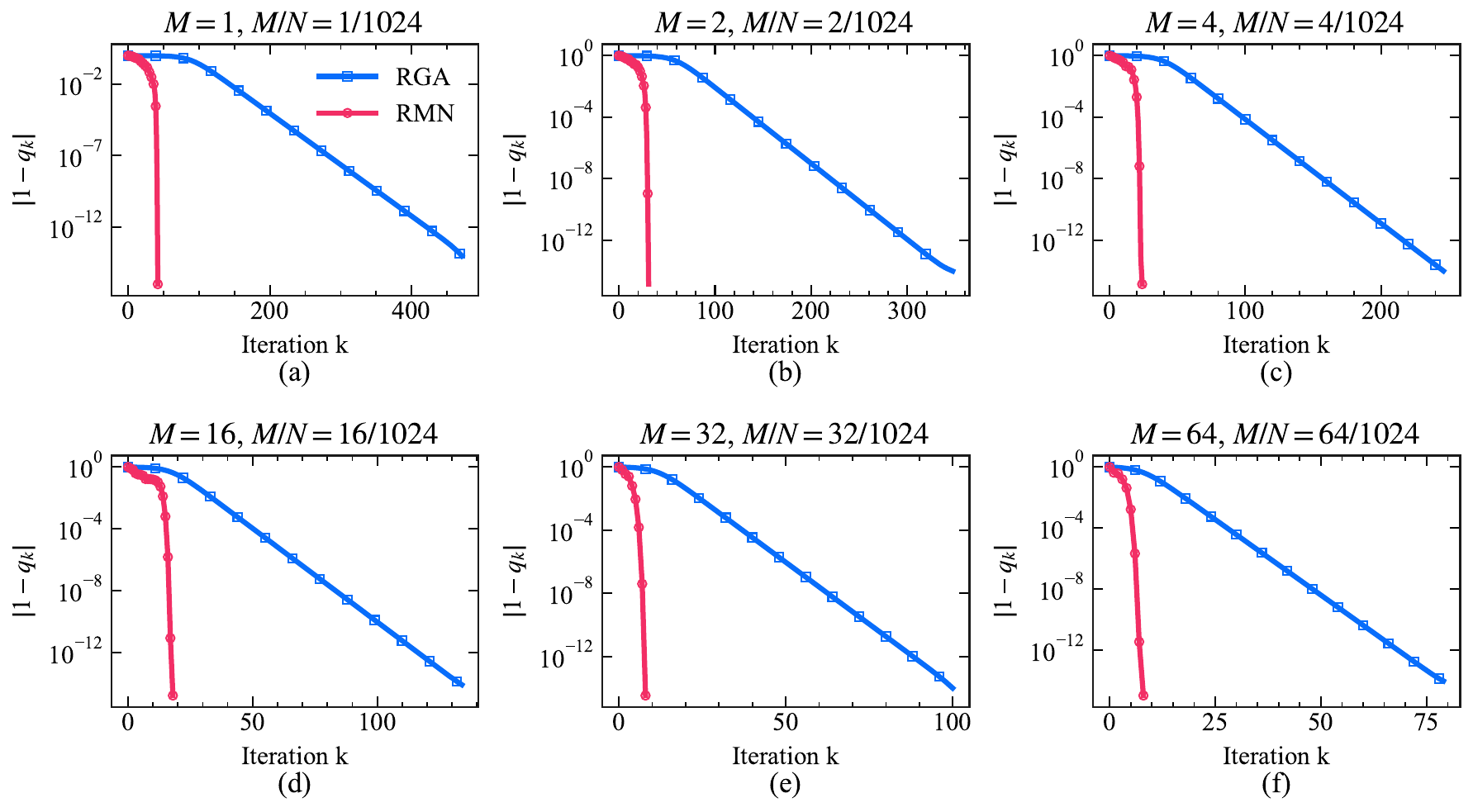}
    \caption{Convergence comparison between the RGA and RMN methods for a fixed problem size \(n=10\) and multiple numbers of marked items \(M=1,2,4,16,32,64\). RGA remains linearly convergent, while RMN retains quadratic convergence, and the iteration count decreases as \(M\) increases.}
    \label{fig:task4_multi_marked_items}
\end{figure*}

\paragraph{Line-search counts and damping sensitivity.}

Finally, we record the total iteration counts and Armijo line search statistics under different damping parameters in RMN. We reset \(M=1\) and run RMN for \(n=10,15\) with damping parameters \(\delta=10^{-2},10^{-3},10^{-4},10^{-5}\). \cref{tab:task5_armijo_delta_sensitivity} reports the total number of RMN iterations, the average and maximum numbers of Armijo backtracking reductions \(m_k\), and the percentage of accepted updates with \(m_k=0\). For \(n=10\), the iteration count remains between \(37\) and \(42\) for all tested \(\delta\), while the average number of backtracking reductions ranges from \(0.69\) to \(3.00\). For \(n=15\), the iteration count remains between \(272\) and \(305\), with average backtracking counts between \(1.51\) and \(4.65\). These values indicate that both the iteration count and the line-search cost remain moderate over the tested damping range.

\begin{table}[!t]
    \centering
    \caption{Line-search counts and damping sensitivity of RMN for \(M=1\) and \(n=10,15\). The table reports iteration counts, Armijo backtracking statistics \(m_k\), and full-step acceptance rates for different \(\delta\).}
    \label{tab:task5_armijo_delta_sensitivity}
    \footnotesize
    \setlength{\tabcolsep}{3pt}
    \begin{ruledtabular}
    \begin{tabular}{@{}cccccc@{}}
        \(n\) & \(\delta\) & Iter. & Avg. \(m_k\) & Max \(m_k\) & \(m_k=0\) \\
        \hline
        10 & \(10^{-2}\) & 42 & 0.69 & 4  & 55\% \\
        10 & \(10^{-3}\) & 42 & 1.19 & 6  & 38\% \\
        10 & \(10^{-4}\) & 38 & 1.74 & 7  & 45\% \\
        10 & \(10^{-5}\) & 37 & 3.00 & 11 & 32\% \\

        15 & \(10^{-2}\) & 305 & 1.51 & 10 & 38\% \\
        15 & \(10^{-3}\) & 292 & 3.28 & 14 & 25\% \\
        15 & \(10^{-4}\) & 278 & 3.59 & 14 & 24\% \\
        15 & \(10^{-5}\) & 272 & 4.65 & 17 & 25\% \\
    \end{tabular}
    \end{ruledtabular}
\end{table}

\section{Discussion}\label{sec-discussion}

In this work, we develop a deeper understanding of the quantum search problem through the lens of the manifold optimization framework by introducing the Grover-compatible Riemannian modified Newton (RMN) method.
While the previous first-order Riemannian gradient ascent (RGA) method \cite{lai2025grover} successfully recasts quantum search as a maximization problem on the unitary group and achieves $\mathcal{O}(\sqrt{N/M}\log(1/\varepsilon))$ query complexity, it is limited to a linear convergence rate with respect to the target accuracy $\varepsilon$. By incorporating second-order geometric information, we overcome this limitation.

The cornerstone of our RMN method is the theoretical observation that, for projector Hamiltonians $H=H^2$, the Riemannian gradient is an eigenvector of the corresponding Riemannian Hessian. This property eliminates the need for costly Hessian inversion, reducing the Newton update to a scaled gradient step.
Consequently, our RMN method can achieve a quadratic convergence rate, driving the total complexity down to $\mathcal{O}(\sqrt{N/M}+\log\log(1/\varepsilon))$ without incurring additional overhead per iteration.
Furthermore, we demonstrate that the classical simulability of the parameter update process is also preserved, allowing for the efficient precomputation of gate angles prior to execution on quantum hardware.

Our results suggest several promising directions for future research. The exact collinearity of the Riemannian gradient and the Newton direction established in our framework (i.e., \cref{thm-hessian}) relies strictly on the projective nature of the Grover Hamiltonian ($H=H^2$). Therefore, extending this second-order geometric property to more general quantum tasks presents a compelling theoretical challenge.
On the other hand, for broader quantum applications such as ground-state preparation for generic molecular Hamiltonians \cite{peruzzo2014variational,kandala2017hardware}, implementing exact Newton or gradient steps may be intractable on quantum computers. In such settings, it is of interest to investigate whether similar invariant gradient subspaces (as in \cref{thm-grad-in-W}) can be identified to reduce the effective problem dimensionality. Furthermore, the application of Riemannian quasi-Newton methods \cite{huang2015broyden}, particularly Riemannian BFGS \cite{huang2018riemannian}, could offer a viable pathway to achieving superlinear convergence without the need for explicit Hessian evaluations, enabling accelerations across a wider spectrum of Hamiltonians.

Beyond those theoretical generalizations, future work can evaluate the resilience of the RMN method under realistic hardware noise, as high-precision algorithms may be sensitive to gate infidelities and decoherence in NISQ devices. Finally, integrating this manifold optimization framework with quantum error mitigation techniques may be essential to stabilize its quadratic convergence for practical implementations.

\begin{acknowledgments}
This work was supported by the National Natural Science Foundation of China under the grant numbers 12501419, 12331010 and 12288101, the National Key R\&D Program of China under the grant number 2024YFA1012901, the Quantum Science and Technology-National Science and Technology Major Project via Project 2024ZD0301900, and the Fundamental Research Funds for the Central Universities, Peking University.
We thank Zongqi Wan for their valuable feedback on the manuscript.
\end{acknowledgments}

\section*{Data availability}

The source code of the numerical experiments is publicly available \cite{GroverNewton}.

\appendix

\section{Geometric tools on \texorpdfstring{$\mathrm{U}(N)$}{U(N)}}\label{app-geo}

The set of $N \times N$ unitary matrices,
\begin{equation*}
\mathrm{U}(N)=\left\{U \in \mathbb{C}^{N \times N} \mid U^{\dagger} U=I\right\},
\end{equation*}
forms a Riemannian manifold, despite being a subset of a Euclidean space.
As a result, optimization over $\mathrm{U}(N)$ requires geometric tools adapted to the manifold structure.
This appendix briefly reviews the basic ingredients for optimization on $\mathrm{U}(N)$, including tangent spaces, Riemannian gradients, retractions, and Riemannian Hessians. For a comprehensive treatment, see \cite{absil2008optimization,boumal2023introduction}.

\subsection{Tangent space}

Geometrically, the \textit{tangent space} provides a local linear approximation of the manifold at a given point, analogous to a tangent plane to a two-dimensional sphere (see \cref{fig:gradient}). From an optimization perspective, it consists of all feasible directions (i.e., tangent vectors) along curves on the manifold. This structure can also be characterized algebraically.

Fix a point $U \in \mathrm{U}(N)$, and consider a smooth curve $t \mapsto U(t) \in \mathrm{U}(N)$ with $U(0)=U$. Differentiating the unitarity constraint $U(t) U(t)^{\dagger}=I$ with respect to $t$ at $t=0$ yields
\begin{equation*}
    U \dot{U}(0)^{\dagger}+\dot{U}(0) U^{\dagger}=0.
\end{equation*}
Letting the tangent vector be $\eta:=\dot{U}(0)$, the above condition implies $(\eta U^{\dagger})^{\dagger}=-(\eta U^{\dagger}).$ Hence, the matrix $\Omega:=\eta U^{\dagger}$ is skew-Hermitian.
Consequently, the tangent space at $U$ is defined as
\begin{equation*}
    T_U=\left\{\Omega U: \Omega^{\dagger}=-\Omega\right\}=\mathfrak{u}(N) U.
\end{equation*}
Here, $\mathfrak{u}(N)$ denotes the Lie algebra of the unitary group, consisting of all $N \times N$ skew-Hermitian matrices and closed under the Lie bracket $[X, Y]=X Y-Y X$. In particular, at the identity $I \in \mathrm{U}(N)$, we have $T_I=\mathfrak{u}(N)$. Since $\operatorname{dim} \mathfrak{u}(N)=N^2$, the manifold $\mathrm{U}(N)$ has dimension $N^2$. Note that the tangent spaces of the unitary group differ only through right multiplication by $U$.

\begin{figure}
    \centering
    \includegraphics[width=0.9\linewidth]{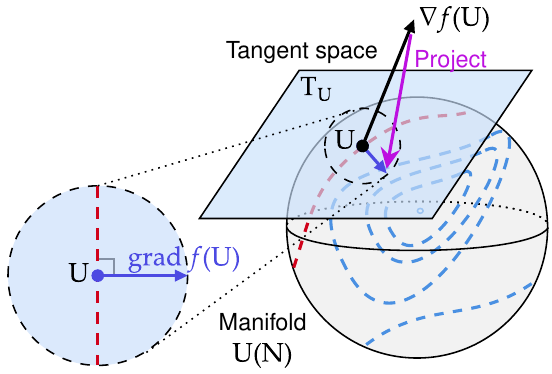}
    \caption{Riemannian gradient on the manifold. Viewing $\mathrm{U}(N)$ as a sphere for illustration, the Riemannian gradient $\operatorname{grad} f(U)$ is the orthogonal projection of the Euclidean gradient $\nabla f(U)$ onto the tangent space at $U$. It is perpendicular to the contour and points in the direction of the fastest increase of $f$.}
    \label{fig:gradient}
\end{figure}

\subsection{Riemannian gradient}

Consider a function $f: \mathrm{U}(N) \subseteq \mathbb{C}^{N \times N} \rightarrow \mathbb{R}$. The Euclidean gradient $\nabla f(U)$ gives the direction of steepest ascent in the ambient space $\mathbb{C}^{N \times N}$, but it generally does not lie in the tangent space $T_U$ and thus is not a feasible direction on the manifold.
The standard approach is to orthogonally project the Euclidean gradient onto $T_U$ (see \cite[Proposition 3.61]{boumal2023introduction}); the resulting tangent vector is called the \textit{Riemannian gradient}, $\grad f(U) \in T_U$.
Geometrically, the Riemannian gradient is perpendicular to the contour line of $f$ passing through $U$ (see \cref{fig:gradient}).

Here, we equip all tangent spaces $T_U$ with the Frobenius inner product $\langle A, B\rangle=\operatorname{Tr}\left(A^{\dagger} B\right),$ which is the metric inherited from the ambient space $\mathbb{C}^{N \times N}$.

The normal space at $U \in \mathrm{U}(N)$ is defined as the orthogonal complement of the tangent space:
\begin{equation*}
    N_U:=\left\{Z \in \mathbb{C}^{N \times N}:\langle Z, X\rangle=0, \forall X \in T_U\right\}.
\end{equation*}
Algebraically, this space can be characterized as
\begin{equation*}
    N_U=\left\{M U: M^{\dagger}=M\right\}=\mathcal{H}(N) U,
\end{equation*}
where $\mathcal{H}(N)$ is the space of $N \times N$ Hermitian matrices.

Using the orthogonal decomposition $\mathbb{C}^{N \times N}=T_U \oplus N_U$, any $Z \in \mathbb{C}^{N \times N}$ can be written as $Z=(\Omega+M) U,$ where $\Omega^{\dagger}=-\Omega$ and $M^{\dagger}=M$. Multiplying by $U^{\dagger}$ on the right gives $Z U^{\dagger}=\Omega+M.$ Taking the skew-Hermitian part yields $\Omega=\operatorname{Skew}\left(Z U^{\dagger}\right)$ with $\operatorname{Skew}(A):= \frac{1}{2}(A - A^\dagger)$. Hence, the orthogonal projection of $Z$ onto $T_U$ is
\begin{equation}\label{eq-projection}
    \mathcal{P}_U(Z):= \operatorname{Skew}(Z U^\dagger) U.
\end{equation}
The Riemannian gradient is then $\operatorname{grad} f(U)=\mathcal{P}_U(\nabla f(U))$.

For the cost function $f(U)=\operatorname{Tr}\left(H U \psi_0 U^{\dagger}\right)$, the Euclidean gradient is $\nabla f(U)=2 H U \psi_0$. In general, if $A$ and $B$ are (skew-)Hermitian, then $\operatorname{Skew}(2 A B)=[A, B].$ Applying this identity with $A:=H$ and $B:=U \psi_0 U^{\dagger}$, and using \cref{eq-projection}, we obtain
\begin{equation*}
    \grad f (U)
    =\Skew (2 H U \psi_0 U^\dagger)U
    =[H, U \psi_0 U^\dagger]U,
\end{equation*}
We denote the left skew-Hermitian factor by $\widetilde{\operatorname{grad}} f(U):= [H, U \psi_0 U^\dagger] \in \mathfrak{u}(N)$.

\subsection{Retractions}

Standard additive updates $U+t \eta$ leave the unitary manifold $\mathrm{U}(N)$, as they violate the unitarity constraint. To map such updates back onto $\mathrm{U}(N)$ while preserving the search direction, one introduces a geometric mapping called a \textit{retraction} \cite{adler2002newton,absil2008optimization}, which is a central tool in the modern manifold optimization framework.

Define $T \mathrm{U}(N):=\left\{(U, \eta): U \in \mathrm{U}(N), \eta \in T_U\right\}.$
A retraction on $\mathrm{U}(N)$ is a smooth mapping
\begin{equation*}
    \mathrm{R}: T \mathrm{U}(N) \rightarrow \mathrm{U}(N), \quad(U, \eta) \mapsto \mathrm{R}_U(\eta),
\end{equation*}
such that for any $U \in \mathrm{U}(N)$ and $\eta \in T_U$, the induced curve $t \mapsto \gamma(t):=\mathrm{R}_U(t \eta)$ satisfies
\begin{equation*}
    \gamma (0)=U \quad \text { and } \quad \dot{\gamma} (0)=\eta.
\end{equation*}
For the unitary manifold $\mathrm{U}(N)$, several retractions are available. Given $U \in \mathrm{U}(N)$ and a tangent vector $\eta=\tilde{\eta} U \in T_U$ with $\tilde{\eta} \in \mathfrak{u}(N)$, common retractions include:

\begin{enumerate}
    \item Riemannian exponential map \cite{hall2015liegroups}:
\begin{equation}\label{eq-app-exp}
    \mathrm{R}_U^{\text{exp}}(\eta) = \exp(\widetilde{\eta})\, U,
\end{equation}
which generates the exact geodesic curves on $\mathrm{U}(N)$.

    \item (First-order) Trotter retraction \cite{trotter1959product,lloyd1996universal,lai2026quantum}:
\begin{equation}\label{eq-trt}
    \mathrm{R}_U^{\text{trt}}(\eta)
    = \left( \prod_{j=1}^{N^2} \exp\!\left( i \omega^j P^j \right) \right) U,
\end{equation}
where $\widetilde{\eta} \in \mathfrak{u}(N)$ is expanded in the $\log_2(N)$-qubit Pauli basis $\{P^j\}_{j=1}^{N^2}$ as
$\widetilde{\eta} = i \sum_j \omega^j P^j$ for $\omega^j \in \mathbb{R}$.

    \item Cayley transform retraction \cite{wen2013feasible}:
\begin{equation*}
    \mathrm{R}_U^{\text{cay}}(\eta)
    = \left(I - \tfrac{1}{2}\widetilde{\eta}\right)^{-1}
      \left(I + \tfrac{1}{2}\widetilde{\eta}\right) U.
\end{equation*}

    \item QR decomposition retraction \cite{absil2008optimization}:
\begin{equation*}
    \mathrm{R}_U^{\text{qr}}(\eta) = \operatorname{qf}(U + \eta),
\end{equation*}
where $\operatorname{qf}(A)$ denotes the orthogonal $Q$-factor in the QR decomposition of $A$.

    \item Polar decomposition retraction \cite{absil2008optimization}:
\begin{equation*}
    \mathrm{R}_U^{\text{polar}}(\eta)
    = (U+\eta)(I + \eta^\dagger \eta)^{-1/2},
\end{equation*}
where $(I+\eta^{\dagger}\eta)^{-1/2}$ is the inverse of the unique positive definite Hermitian square root of $I+\eta^{\dagger}\eta$.
\end{enumerate}

The effectiveness of these retractions strongly depends on the computational paradigm. The Riemannian exponential map and the Trotter retraction correspond to physical Hamiltonian evolution, making them naturally suited for implementation on quantum hardware; however, their reliance on matrix exponentials renders them computationally inefficient on classical computers.
Conversely, the Cayley, QR, and polar retractions rely on matrix inversions and factorizations. While implementing these nonlinear operations on quantum circuits incurs additional overhead, they are highly optimized within numerical linear algebra, making them the preferred choice in classical computing environments.

\subsection{Riemannian Hessian}

Unlike the usual setting, where the Hessian is represented by a symmetric matrix of second-order partial derivatives, the Riemannian Hessian at \(U\), denoted by \(\operatorname{Hess} f(U)\), is a self-adjoint linear operator on \(T_U\). It describes the covariant derivative of the Riemannian gradient field along a tangent direction, thereby naturally incorporating both the second-order variation of the objective function and the curvature of the manifold.

A rigorous derivation of \(\operatorname{Hess} f(U)\) requires second-order differential geometric tools that are beyond the scope of this paper. Here we present only the main computation, based on the simplification afforded by the embedded submanifold setting in \cite[Corollary 5.16]{boumal2023introduction}: one differentiates a smooth ambient extension of the Riemannian gradient in \(\mathbb C^{N\times N}\), and then projects the result back onto the tangent space \(T_U\).

For \(f(U)=\operatorname{Tr}(HU\psi_0U^\dagger)\), set \(\psi:=U\psi_0U^\dagger\). The Riemannian gradient is \(\operatorname{grad} f(U)=[H,\psi]U\). Consider its ambient extension
\[
        \mathcal G:\mathbb C^{N\times N}\to\mathbb C^{N\times N},
        \qquad
        \mathcal G(U):=[H,U\psi_0U^\dagger]U .
\]
A direct differentiation gives
\[
        \mathrm D\mathcal G(U)[V]=[H,V\psi_0U^\dagger+U\psi_0V^\dagger]U+[H,U\psi_0U^\dagger]V.
\]
Restricting this derivative to a tangent direction \(V=\Omega U\in T_U\), with \(\Omega\in\mathfrak u(N)\), yields
\[
        \mathrm D\mathcal G(U)[\Omega U]
        =
        \bigl([H,[\Omega,\psi]]+[H,\psi]\Omega\bigr)U.
\]
Since this ambient derivative is not necessarily tangent to \(\mathrm U(N)\), we project it onto \(T_U\). By \cref{eq-projection},
\begin{align*}
        \operatorname{Hess} f(U)[\Omega U]
        &=\mathcal P_U\bigl(\mathrm D\mathcal G(U)[\Omega U]\bigr)\\
        &=\operatorname{Skew}\bigl(\mathrm D\mathcal G(U)[\Omega U]U^\dagger\bigr)U\\
        &=\operatorname{Skew}\bigl([H,[\Omega,\psi]]+[H,\psi]\Omega\bigr)U.
\end{align*}
Here \([H,[\Omega,\psi]]\) is skew-Hermitian, while \([H,\psi]\) and \(\Omega\) are both skew-Hermitian. Hence, we obtain
\begin{align*}
\operatorname{Hess} f(U)[\Omega U] & =\left([H,[\Omega, \psi]]+\frac{1}{2}[[H, \psi], \Omega]\right) U \\
& =\frac{1}{2}([H,[\Omega, \psi]]+[[H, \Omega], \psi]) U,
\end{align*}
where we used the Jacobi identity \( [[H,\psi],\Omega]=-[H,[\Omega,\psi]]+[[H,\Omega],\psi]\). This is the Hessian formula used in \cref{lem-hessian}.

In general manifold optimization, computing a Newton step necessitates the computationally expensive inversion of this Hessian operator. However, as demonstrated in \cref{thm-hessian}, when the Hamiltonian is a projector \((H=H^2)\), the Riemannian gradient strictly emerges as an exact eigenvector of this Hessian operator.

\section{Complexity analysis for the Grover-compatible RMN method}\label{app-complexity}

This section proves the convergence and iteration complexity of \cref{alg-newton-armijo}, the Grover-compatible Riemannian modified Newton (RMN) method.
We use the notation introduced in the main text:
\[
\begin{alignedat}{2}
        f(U)&=\operatorname{Tr}(HU\psi_0U^\dagger), \qquad\qquad&
        \psi_k&=U_k\psi_0U_k^\dagger,\\
        q_k&=f(U_k),&
        \mathrm g_k&=[H,\psi_k],\\
        G_k&:=\|\mathrm g_k\|_F^2=2q_k(1-q_k),&
        e_k&:=1-q_k .
\end{alignedat}
\]
The RMN direction used in \cref{alg-newton-armijo} is $\gamma_k \mathrm g_k U_k,$ where
\[
        \gamma_k=\frac{1}{\max\{\delta,2q_k-1\}},
\]
where \(\delta>0\) is the damping parameter. We also write $q_0=\frac{M}{N}.$ Unless otherwise specified, all results in this section are stated for the \(5\)-factor Grover-compatible retraction $\mathrm R_U\equiv \mathrm R_U^{(5)}$ defined in \cref{eq-5factor-retraction}; the proofs for other forms of Grover-compatible retractions are similar.

The complexity proof for \cref{alg-newton-armijo} has four parts.
First, we prove that the iterates are well defined, monotone, and globally convergent to success probability one.
Second, we identify a Newton region \(e_k\le r_{\mathrm N}\), where the damping is inactive, the full Newton step is accepted, and the error satisfies the quadratic reduction \(e_{k+1}\le C_{\mathrm N}e_k^2\).
Third, we show that the global phase enters this Newton region in \(\mathcal O(\sqrt{N/M})\) iterations.
Finally, combining the entrance estimate with the local quadratic reduction yields the total complexity
\[
        \mathcal O\!\left(
        \sqrt{\frac{N}{M}}+\log\log\frac1\varepsilon
        \right)
\]
for \cref{alg-newton-armijo}, as stated in \cref{thm:global-rmn-complexity-final} at the end of this section.

\subsection{Global monotone convergence}

In \cref{sec-rga}, we revisit the first-order Riemannian gradient ascent (RGA) method, namely \cref{alg-grover-ret}, proposed in \cite{lai2025grover}.
We first state two auxiliary lemmas from \cite{lai2025grover}, which are also important for proving the global convergence of our second-order algorithm; their proofs can be found therein.

\begin{lemma}[First and second order bounds {\cite[Lemma 4.2]{lai2025grover}}]
\label{lem:two-bounds}
For any \(U\in\mathrm{U}(N)\) and any \(\eta\in\mathcal W U\), we have
\begin{equation*}
\begin{aligned}
        \|\mathrm R_U^{(5)}(\eta)-U\|_F
        &\le \|\eta\|_F,\\
        \|\mathrm R_U^{(5)}(\eta)-U-\eta\|_F
        &\le
        \frac{1}{4\sqrt{2q_0(1-q_0)}}\|\eta\|_F^2.
\end{aligned}
\end{equation*}
\end{lemma}

\begin{lemma}[Riemannian Lipschitz constant {\cite[Proposition 4.3]{lai2025grover}}]
Define the constant
\begin{equation}
\label{eq:Lrie-rmn-final}
\begin{aligned}
        L_{\mathrm{Rie}}
        &:=2+\frac{1}{\sqrt{2q_0(1-q_0)}}\\
        &=2+\frac{N}{\sqrt{2M(N-M)}}\\
        &\in\mathcal O(\sqrt{N/M}).
\end{aligned}
\end{equation}
Then, for any \(U\in\mathrm{U}(N)\), the pullback \(f\circ \mathrm R_U^{(5)}:\mathcal W U\to\mathbb R\) satisfies: $\forall \eta \in \mathcal{W} U$,
\begin{equation}
\label{eq:pullback-ascent-revised}
\begin{aligned}
 f\bigl(\mathrm R_U^{(5)}(\eta)\bigr)
\ge
 f(U)+\langle \operatorname{grad} f(U),\eta\rangle
      -\frac{L_{\mathrm{Rie}}}{2}\|\eta\|_F^2.
\end{aligned}
\end{equation}
\end{lemma}

The Riemannian Lipschitz constant \(L_{\mathrm{Rie}}\) is the manifold analogue of the Lipschitz constant used in Euclidean complexity analysis, and the lower bound in \cref{eq:pullback-ascent-revised} is the basic estimate needed in the convergence proof.

\begin{lemma}[Lower bound of the effective step]
Assume that \cref{alg-newton-armijo} is run with \(0<\delta\le1\), \(0<c<1\), and \(0<\rho<1\). At every nonstationary iterate, the Armijo backtracking loop terminates after finitely many reductions. Moreover, if \(t_k\) is the accepted step and \(\gamma_k\) is the corresponding scaling factor, then accepted effective step $t_k\gamma_k$ satisfies
\begin{equation}\label{eq:effective-step-revised}
t_k\gamma_k\ge \sigma:=\min\left\{1,\frac{2\rho(1-c)}{L_{\mathrm{Rie}}}\right\}.
\end{equation}
\end{lemma}

\begin{proof}
Let \(G_k=\|\mathrm g_k\|_F^2\). Since the iterate is nonstationary, \(G_k>0\). Applying \cref{eq:pullback-ascent-revised} with \(\eta=t\gamma_k\mathrm g_kU_k\) gives
\[
f\bigl(\mathrm R_{U_k}^{(5)}(t\gamma_k\mathrm g_kU_k)\bigr)
\ge
q_k+t\gamma_kG_k-\frac{L_{\mathrm{Rie}}}{2}t^2\gamma_k^2G_k.
\]
Thus the Armijo condition is guaranteed if the right-hand side above is at least \(q_k+c t\gamma_kG_k\), namely,
\[
        q_k+t\gamma_kG_k-\frac{L_{\mathrm{Rie}}}{2}t^2\gamma_k^2G_k
        \ge
        q_k+c t\gamma_kG_k.
\]
After subtracting \(q_k\) and dividing by \(t\gamma_kG_k>0\), this sufficient condition becomes $1-\frac{L_{\mathrm{Rie}}}{2}t\gamma_k\ge c,$ or equivalently
\[
        t\le B_k:=\frac{2(1-c)}{L_{\mathrm{Rie}}\gamma_k}.
\]
Therefore every trial step \(t\le B_k\) satisfies the Armijo condition. Since the backtracking trial steps are \(1,\rho,\rho^2,\ldots\) and \(\rho^m\to0\), there exists a finite \(m\) such that \(\rho^m\le B_k\). Hence the Armijo backtracking loop terminates after finitely many reductions.

It remains to prove the lower bound $ \sigma$ given in \cref{eq:effective-step-revised} for the accepted effective step \(t_k\gamma_k\). We consider two cases.
\noindent(Case 1)
First suppose that \(B_k\ge1\). Then $t_k=1$ is accepted. Moreover, since \(\gamma_k:=1/\max\{\delta,2q_k-1\}\), \(0<\delta\le1\), and \(2q_k-1\le1\), we have \(\max\{\delta,2q_k-1\}\in(0,1]\). Hence \(\gamma_k\ge1\). Therefore
\begin{equation}\label{eq-165684}
t_k\gamma_k=\gamma_k\ge1.
\end{equation}
\noindent(Case 2)
Now suppose that \(B_k<1\). Let \(t_k=\rho^m\) be the first accepted trial step. If \(m=0\), then \(t_k=1>\rho B_k\), because \(B_k<1\) and \(0<\rho<1\). If \(m\ge1\), then the previous trial step \(\rho^{m-1}\) was rejected. Since every step \(t\le B_k\) is guaranteed to satisfy Armijo, the rejected step cannot satisfy \(\rho^{m-1}\le B_k\). Hence $\rho^{m-1}>B_k.$
Multiplying this inequality by \(\rho\), we obtain \(t_k=\rho^m>\rho B_k\). Thus, whether \(m=0\) or \(m\ge1\), we have \(t_k>\rho B_k\). Since \(\gamma_k>0\), it follows that
\begin{equation}\label{eq-165685}
t_k\gamma_k>\gamma_k\rho B_k=\frac{2\rho(1-c)}{L_{\mathrm{Rie}}}.
\end{equation}
Combining the two cases, \cref{eq-165684,eq-165685} gives the lower bound in \cref{eq:effective-step-revised}.
\end{proof}

\begin{proposition}[Global monotone convergence of \cref{alg-newton-armijo}]
\label{prop:global-monotone-revised}
Let $\{q_k\}$ be generated by \cref{alg-newton-armijo} with the $5$-factor retraction, $0<\delta\le1$, $0<c\le1/2$, and $0<\rho<1$.  Then
\begin{equation}
\label{eq:global-logistic-revised}
        q_{k+1}
        \ge
        q_k+2c\sigma q_k(1-q_k),
        \qquad k=0,1,2,\ldots,
\end{equation}
where $\sigma$ is defined in \cref{eq:effective-step-revised}.  Consequently, $\{q_k\}$ is monotonically nondecreasing and converges to $1$.
\end{proposition}

\begin{proof}
The accepted Armijo step gives $q_{k+1}\ge q_k+c\,t_k\gamma_kG_k .$
Using \cref{eq:effective-step-revised} and the identity $G_k=2q_k(1-q_k)$, we obtain $q_{k+1} \ge q_k+c\sigma G_k = q_k+2c\sigma q_k(1-q_k),$ which proves \cref{eq:global-logistic-revised}. In particular, since \(c>0\), \(\sigma>0\), and \(q_k(1-q_k)\ge0\), the sequence \(\{q_k\}\) is monotonically nondecreasing.

Since \(q_k=f(U_k)\) is a success probability, we always have \(0\le q_k\le1\). Therefore \(\{q_k\}\) is monotone and bounded above, and hence it has a limit. Denote this limit by
\[
        \bar q:=\lim_{k\to\infty}q_k\in[q_0,1].
\]
We claim that \(\bar q=1\). Suppose, to the contrary, that \(\bar q<1\). Since \(q_0=M/N>0\) and \(\{q_k\}\) is nondecreasing, we have \(q_k\ge q_0>0\) for all \(k\). By \(q_k\to\bar q\) and \(\bar q<1\), there exists \(K\) such that \(|q_k-\bar q|<\varepsilon_0:=(1-\bar q)/2\) for all \(k\ge K\). Thus, for all \(k\ge K\), we have
\[
        q_k<\bar q+\frac{1-\bar q}{2}=\frac{1+\bar q}{2},
        \quad
\text{ and }
        \quad
        1-q_k>\frac{1-\bar q}{2}>0.
\]
Then \cref{eq:global-logistic-revised} implies that, for all \(k\ge K\),
\[
        q_{k+1}-q_k
        \ge
        2c\sigma q_k(1-q_k)
        >
        c\sigma q_0(1-\bar q)>0.
\]
This contradicts the convergence of \(\{q_k\}\), because \(q_{k+1}-q_k\to\bar q-\bar q=0\). Hence \(\bar q=1\).
\end{proof}

\subsection{Newton region and local quadratic convergence}

In this subsection, we will define the Newton region that the global phase must reach.
Its radius is chosen sufficiently small so that the damping in the modified Newton direction is inactive and the full Newton step satisfies the Armijo condition.
Consequently, once \cref{alg-newton-armijo} enters this region, it reduces to the standard Newton method.
The following two lemmas are specific to the projector $H=H^2$.  They are the key identities behind the quadratic convergence result.

\begin{lemma}[Second-order residual of the gradient map]
\label{lem:gradient-map-residual-revised}
Define the gradient map
\begin{equation*}
    \mathcal G(U):=[H,U\psi_0U^\dagger].
\end{equation*}
For $\eta=\Omega U$ with $\Omega\in\mathcal W$, its Euclidean directional derivative is
\begin{equation}
\label{eq:G-derivative-revised}
        D\mathcal G(U)[\eta]=[H,[\Omega,\psi_U]],
        \qquad \psi_U:=U\psi_0U^\dagger.
\end{equation}
Moreover, it satisfies
\begin{equation}
\label{eq:G-residual-revised}
        \|
        \mathcal G\bigl(\mathrm R_U^{(5)}(\eta)\bigr)
        -\mathcal G(U)-D\mathcal G(U)[\eta]
        \|_F
        \le
        L_{\mathrm{Rie}}\|\eta\|_F^2.
\end{equation}
\end{lemma}

\begin{proof}
We first compute the Euclidean directional derivative. Since \(\eta=\Omega U\) and \(\Omega^\dagger=-\Omega\), we have
\begin{equation*}
\begin{aligned}
&\frac{d}{dt}\bigg|_{t=0}
        (U+t\eta)\psi_0(U+t\eta)^\dagger\\
&\quad
        =\eta\psi_0U^\dagger+U\psi_0\eta^\dagger\\
&\quad
        =\Omega\psi_U-\psi_U\Omega\\
&\quad
        =[\Omega,\psi_U].
\end{aligned}
\end{equation*}
Taking the commutator with \(H\) gives \cref{eq:G-derivative-revised}. It remains to prove the second-order residual bound in \cref{eq:G-residual-revised}. It remains to prove the second-order residual bound. Set $W:=\mathrm R_U^{(5)}(\eta)$, $\Delta:=W-U$, and $E:=\Delta-\eta=W-U-\eta$. The two estimates for the $5$-factor retraction in \cref{lem:two-bounds} give
\begin{equation}
\label{eq:Delta-E-revised}
        \|\Delta\|_F\le\|\eta\|_F,
        \qquad
        \|E\|_F\le
        \frac{1}{4\sqrt{2q_0(1-q_0)}}\|\eta\|_F^2.
\end{equation}
Since $W=U+\Delta$ and $\Delta=\eta+E$, subtracting the linear state variation $\eta\psi_0U^\dagger+U\psi_0\eta^\dagger$ from the identity
\begin{equation*}
\begin{aligned}
W\psi_0W^\dagger
&=(U+\Delta)\psi_0(U+\Delta)^\dagger\\
&=U\psi_0U^\dagger+\Delta\psi_0U^\dagger\\
&\quad
  +U\psi_0\Delta^\dagger+\Delta\psi_0\Delta^\dagger
\end{aligned}
\end{equation*}
yields
\begin{equation*}
\begin{aligned}
&W\psi_0W^\dagger-U\psi_0U^\dagger
        -\bigl(\eta\psi_0U^\dagger+U\psi_0\eta^\dagger\bigr)\\
&\quad
=E\psi_0U^\dagger+U\psi_0E^\dagger+\Delta\psi_0\Delta^\dagger .
\end{aligned}
\end{equation*}
We now apply the commutator with \(H\). Since \(D\mathcal G(U)[\eta]=[H,\eta\psi_0U^\dagger+U\psi_0\eta^\dagger]\),
\begin{align}
\label{eq:G-residual-expansion-revised}
&        \mathcal G(W)-\mathcal G(U)-D\mathcal G(U)[\eta] \notag\\
&\quad =
        [H,W\psi_0W^\dagger]
        -[H,U\psi_0U^\dagger]
        -[H,\eta\psi_0U^\dagger+U\psi_0\eta^\dagger] \notag\\
&\quad =
        [H,W\psi_0W^\dagger-U\psi_0U^\dagger
        -(\eta\psi_0U^\dagger+U\psi_0\eta^\dagger)] \notag\\
&\quad =
        [H,E\psi_0U^\dagger+U\psi_0E^\dagger+\Delta\psi_0\Delta^\dagger] \notag\\
&\quad =
        [H,E\psi_0U^\dagger+U\psi_0E^\dagger]
        +[H,\Delta\psi_0\Delta^\dagger].
\end{align}

Since \(H\) is an orthogonal projector, it spectral norm \(\|H\|_2\le1\). Hence, for any matrix \(A\), the mixed-norm inequalities give
\begin{equation*}
\begin{aligned}
\|[H,A]\|_F
&= \|HA-AH\|_F\\
&\le \|HA\|_F+\|AH\|_F\\
&\le \|H\|_2\|A\|_F+\|A\|_F\|H\|_2\\
&\le 2\|A\|_F.
\end{aligned}
\end{equation*}
Also, using $\|\psi_0\|_2=\|U\|_2=1$, we have $\|E\psi_0U^\dagger+U\psi_0E^\dagger\|_F\le2\|E\|_F$, while $\|\Delta\psi_0\Delta^\dagger\|_F\le\|\Delta\|_F^2$. Combining these estimates with \cref{eq:G-residual-expansion-revised} gives
\[
\begin{aligned}
\left\|
        \mathcal G(W)-\mathcal G(U)-D\mathcal G(U)[\eta]
\right\|_F
&\le
        4\|E\|_F+2\|\Delta\|_F^2
\\
&\le
        \bigl(\tfrac{1}{\sqrt{2q_0(1-q_0)}}+2\bigr)\|\eta\|_F^2
\\
&=
        L_{\mathrm{Rie}}\|\eta\|_F^2,
\end{aligned}
\]
where the last step uses \cref{eq:Delta-E-revised} and the definition of $L_{\mathrm{Rie}}$ in \cref{eq:Lrie-rmn-final}. This proves \cref{eq:G-residual-revised}.

\end{proof}

\begin{lemma}[Newton cancellation]
\label{lem:newton-cancellation-revised}
Let $\psi=|\psi\rangle\langle\psi|$ be a pure state,
$q=\operatorname{Tr}(H\psi)$, and $g=[H,\psi]$.  Then
\begin{equation}
\label{eq:commutator-eigen-revised}
        [H,[g,\psi]]=(1-2q)g.
\end{equation}
Consequently, if $q>1/2$ and $\gamma=(2q-1)^{-1}$, then, for
$\eta=\gamma gU$ with $\psi=U\psi_0U^\dagger$,
\begin{equation*}
        \mathcal G(U)+D\mathcal G(U)[\eta]=0.
\end{equation*}
\end{lemma}

\begin{proof}
\cref{eq:commutator-eigen-revised} follows directly from the proof of \cref{thm-hessian}, where we showed that
$L(g)=[[H,g],\psi]=(1-2q)g$ for the linear map $L(\Omega):= \tfrac{1}{2}\bigl([H, [\Omega, \psi]] + [[H, \Omega], \psi]\bigr)$.
If \(q>1/2\) and \(\gamma=(2q-1)^{-1}\), then \cref{eq:G-derivative-revised,eq:commutator-eigen-revised} give
\[
        D\mathcal G(U)[\gamma gU]=[H,[\gamma g,\psi]]=\gamma(1-2q)g=-g=-\mathcal G(U).
\]
This proves the cancellation identity.
\end{proof}

The next result identifies the local Newton phase. We use the radius $r_{\mathrm N}:=1/(128L_{\mathrm{Rie}}^2)$ to define the Newton region \(e_k\le r_{\mathrm N}\). Once the iteration enters this region, it stays there, the damping is inactive, and the Armijo line search accepts the full Newton step.

\begin{proposition}[Newton region and local quadratic convergence of \cref{alg-newton-armijo}]
\label{thm:local-quadratic-revised}
Let \(\{q_k\}\) be generated by \cref{alg-newton-armijo} with the \(5\)-factor retraction, where \(0<\delta\le1/2\), \(0<c\le1/8\), and \(0<\rho<1\). Define the constants
\begin{equation}
\label{eq:CN-rN-rmn-final}
        C_{\mathrm N}:=64L_{\mathrm{Rie}}^2,
        \qquad
        r_{\mathrm N}:=\frac{1}{128L_{\mathrm{Rie}}^2}=\frac{1}{2C_{\mathrm N}}.
\end{equation}
If an iterate enters this region, namely
\[
        e_k:=1-q_k\le r_{\mathrm N},
\]
then either \(e_k=0\), in which case a maximizer has already been reached, or the Armijo line search accepts the full step \(t_k=1\) and
\begin{equation}
\label{eq:quadratic-recurrence-revised}
        e_{k+1}\le C_{\mathrm N}e_k^2\le\frac12 e_k.
\end{equation}
Consequently, \(e_{k+\ell}\le r_{\mathrm N}\) for every \(\ell\ge0\), and the RMN iteration is quadratically convergent in the success probability error.
\end{proposition}

\begin{proof}
By \cref{eq:CN-rN-rmn-final} and \(L_{\mathrm{Rie}}>2\), we have \(r_{\mathrm N}=1/(128L_{\mathrm{Rie}}^2)<1/512\). Thus \(e_k\le r_{\mathrm N}\) implies \(e_k<1/512\), and hence \(q_k=1-e_k>511/512>3/4\). Therefore \(2q_k-1>1/2\ge\delta\). In this region, the damping in \cref{alg-newton-armijo} is inactive, so the modified Newton scaling becomes
\[
        \gamma_k=\frac{1}{\max\{\delta,2q_k-1\}}=\frac{1}{2q_k-1}\le2.
\]
At this point, $\gamma_k\mathrm g_k$ is the standard Newton direction. If \(e_k=0\), then \(q_k=1\), so the current iterate is already a global maximizer. We therefore assume \(e_k>0\) in the rest of the proof.

The following proof is organized around the full Newton candidate:
\[
\begin{aligned}
        U_k^+&:=\mathrm R_{U_k}^{(5)}(\gamma_k\mathrm g_kU_k),\\
        q_k^+&:=f(U_k^+),\\
        e_k^+&:=1-q_k^+,
        \qquad
        \mathrm g_k^+:=\mathcal G(U_k^+).
\end{aligned}
\]
We first show that the norm of gradient at \(U_k^+\) is of order \(G_k\), then use \(q_k^+>1/2\) to convert this into a quadratic bound for the error \(e_k^+=1-f(U_k^+)\), and finally verify that this candidate satisfies the Armijo condition with \(t_k=1\).

\medskip
\noindent {(Step 1: Show that the gradient at the full Newton trial point is of second order.)}
We apply \cref{lem:gradient-map-residual-revised} with \(U=U_k\) and \(\eta=\gamma_k\mathrm g_kU_k\). Since \(q_k>1/2\) and \(\gamma_k=(2q_k-1)^{-1}\), the Newton cancellation identity in \cref{lem:newton-cancellation-revised} gives $\mathcal G(U_k)+D\mathcal G(U_k)[\gamma_k\mathrm g_kU_k]=0.$ Consequently, \cref{lem:gradient-map-residual-revised} implies
\begin{equation}
\label{eq:gplus-bound-revised}
        \|\mathrm g_k^+\|_F
        \le
        L_{\mathrm{Rie}}\|\gamma_k\mathrm g_kU_k\|_F^2
        =
        L_{\mathrm{Rie}}\gamma_k^2\|\mathrm g_k\|_F^2
        \le
        4L_{\mathrm{Rie}}G_k,
\end{equation}
Here we used \(\|\mathrm g_kU_k\|_F=\|\mathrm g_k\|_F\), \(G_k=\|\mathrm g_k\|_F^2\), and \(\gamma_k\le2\). Thus Step~1 gives the key estimate that the gradient at the full Newton trial point is of second order in the current gradient, namely \(\|\mathrm g_k^+\|_F=O(\|\mathrm g_k\|_F^2)\).

\medskip
\noindent {(Step 2: Show that the full Newton candidate still satisfies \(q_k^+>1/2\).)}
The estimates in \cref{lem:two-bounds} give
\[
        \|U_k^+-U_k\|_F
        \le
        \|\gamma_k\mathrm g_kU_k\|_F
        =
        \gamma_k\|\mathrm g_k\|_F
        \le
        2\|\mathrm g_k\|_F.
\]
We next bound the corresponding change in the cost function. Let \(Z(s):=U_k+s(U_k^+-U_k)\), \(s\in[0,1]\), be the Euclidean line segment joining \(U_k\) and \(U_k^+\). Although \(Z(s)\) is not necessarily unitary for \(s\in(0,1)\), it satisfies
\[
        \|Z(s)\|_2
        \le
        (1-s)\|U_k\|_2+s\|U_k^+\|_2
        =
        1,
\]
because both \(U_k\) and \(U_k^+\) are unitary. For the Euclidean extension of \(f\), one has \(\nabla f(U)=2HU\psi_0\). Hence, along this segment, $\|\nabla f(Z(s))\|_F = \|2HZ(s)\psi_0\|_F \le 2\|H\|_2\|Z(s)\|_2\|\psi_0\|_F \le 2.$ By the fundamental theorem of calculus and the Cauchy--Schwarz inequality,
\[
\begin{aligned}
        |q_k^+-q_k|
        &=
        |f(U_k^+)-f(U_k)| \\
        &=
        \left|
        \int_0^1
        \left\langle \nabla f(Z(s)),\,U_k^+-U_k\right\rangle_F\,ds
        \right| \\
        &\le
        \int_0^1
        \|\nabla f(Z(s))\|_F\|U_k^+-U_k\|_F\,ds \\
        &\le
        2\|U_k^+-U_k\|_F
        \le
        4\|\mathrm g_k\|_F.
\end{aligned}
\]
Using \(G_k=\|\mathrm g_k\|_F^2=2q_ke_k\le2e_k\), we obtain $|q_k^+-q_k| \le 4\sqrt{2e_k}.$
Since \(e_k\le r_{\mathrm N}<1/512\), we have \(4\sqrt{2e_k}<1/4\). Together with \(q_k>3/4\), this yields
\[
        q_k^+
        \ge
        q_k-|q_k^+-q_k|
        >
        \frac34-\frac14
        =
        \frac12.
\]
Thus the full Newton candidate stays in the region where \(q>1/2\).

\medskip
\noindent {(Step 3: Show that the error at the full Newton candidate is quadratic in \(e_k\).)}
At \(U_k^+\), the identity \(\|\mathrm g_k^+\|_F^2=2q_k^+(1-q_k^+)=2q_k^+e_k^+\) and the fact \(q_k^+>1/2\) imply
\[
        e_k^+\le \|\mathrm g_k^+\|_F^2 .
\]
Combining this with \cref{eq:gplus-bound-revised} and \(G_k=2q_ke_k\le2e_k\), we obtain
\[
        e_k^+
        \le
        \|\mathrm g_k^+\|_F^2
        \le
        16L_{\mathrm{Rie}}^2G_k^2
        \le
        64L_{\mathrm{Rie}}^2e_k^2
        =
        C_{\mathrm N}e_k^2,
\]
where \(C_{\mathrm N}:=64L_{\mathrm{Rie}}^2\) is defined in \cref{eq:CN-rN-rmn-final}. Since \(e_k\le r_{\mathrm N}=1/(2C_{\mathrm N})\), this also gives
\begin{equation}\label{eq:full-step-error-halving-revised}
        e_k^+
        \le
        C_{\mathrm N}e_k^2
        \le
        \frac12 e_k .
\end{equation}
Thus the full Newton candidate has quadratic error bound and reduces the current error by at least a factor of two.

\medskip
\noindent {(Step 4: Show that the full Newton candidate satisfies the Armijo condition.)}
It remains to show that \cref{alg-newton-armijo} actually accepts the candidate analyzed above.
By \cref{eq:full-step-error-halving-revised}, the full-step increase satisfies
\begin{equation}\label{eq:full-step-increase-lower-revised}
q_k^+-q_k
        =
        (1-e_k^+)-(1-e_k)
        =
        e_k-e_k^+
        \ge
        \frac12 e_k.
\end{equation}
On the other hand, the Armijo increase required for the full step \(t=1\) satisfies
\begin{equation}\label{eq:armijo-demand-upper-revised}
c\gamma_kG_k
        =
        c\gamma_k\,2q_ke_k
        \le
        4ce_k
        \le
        \frac12 e_k,
\end{equation}
where we used \(\gamma_k\le2\), \(q_k\le1\), and \(c\le1/8\). Therefore, \cref{eq:full-step-increase-lower-revised,eq:armijo-demand-upper-revised} yield $q_k^+-q_k\ge c\gamma_kG_k,$ or equivalently,
\[
        q_k^+\ge q_k+c\gamma_kG_k.
\]
This is precisely the Armijo condition with \(t_k=1\). Hence \cref{alg-newton-armijo} accepts the full step, so \(U_{k+1}=U_k^+\), \(e_{k+1}=e_k^+\), and $e_{k+1} \leq C_{\mathrm N} e_k^2$ in \cref{eq:quadratic-recurrence-revised} follows from \cref{eq:full-step-error-halving-revised} in Step~3.
Finally, since \(e_{k+1}\le e_k/2\le r_{\mathrm N}\), the next iterate remains in the Newton region defined by \cref{eq:CN-rN-rmn-final}. Thus the same argument applies inductively, proving local quadratic convergence in the success probability error.
\end{proof}

\subsection{Entrance into the Newton region during the global phase}

Having identified the Newton region \(e_k\le r_{\mathrm N}\) and proved local quadratic convergence inside it, we now return to the global phase and show that \cref{alg-newton-armijo} reaches this region within \(\mathcal O(\sqrt{N/M})\) iterations.
The proof mainly relies on the dynamics in the two-dimensional Grover plane $\mathcal{S}=\operatorname{span}_{\mathbb{C}}\left\{\left|\psi_0\right\rangle, H\left|\psi_0\right\rangle\right\}$.

\begin{remark}[Why a standard first-order optimization route is not used]
Although a standard first-order optimization argument can prove global convergence, it gives a weaker bound for entering the Newton region.
Indeed, such an argument gives the usual progress estimate
\[
        q_{k+1}-q_k\ge \Theta(L_{\mathrm{Rie}}^{-1})\left\|\mathrm{g}_k\right\|_F^2=\Theta(L_{\mathrm{Rie}}^{-1})q_k(1-q_k),
\]
and therefore requires
\[
        \mathcal O\!\left(
        L_{\mathrm{Rie}}
        \left[
        \log\frac1{q_0}+\log\frac1{\varepsilon}
        \right]
        \right)
\]
iterations to reduce the success probability error to a prescribed tolerance \(\varepsilon\), which corresponds to a linear convergence rate; see \cite[Theorem~4.6]{lai2025grover}.
For entrance into the Newton region, the tolerance is chosen as
\[
        \varepsilon_{\mathrm{enter}}=r_{\mathrm N}=\frac{1}{128L_{\mathrm{Rie}}^2}.
\]
Since \(q_0=M/N\) and \(L_{\mathrm{Rie}}\in\mathcal O(\sqrt{N/M})\), this route gives only
\[
        \mathcal O\!\left(\sqrt{\frac NM}\log\frac NM\right).
\]
The extra logarithmic factor cannot be removed by this standard argument.
To recover the Grover scaling \(\mathcal O(\sqrt{N/M})\), the proof below uses the two-dimensional Grover-plane dynamics and measures progress in the Grover angle.
\end{remark}

We first introduce the scalar quantities used only in this subsection.
Recall that \(q_k=f(U_k)\) and \(e_k=1-q_k\), and define the Grover angle by
\begin{equation}\label{eq:grover-angle-enter}
\begin{aligned}
        \vartheta_k&:=\arcsin\sqrt{q_k}\in(0,\pi/2],\\
        q_k&=\sin^2\vartheta_k,
        \qquad
        e_k=\cos^2\vartheta_k .
\end{aligned}
\end{equation}
Since \(\{q_k\}\) is nondecreasing by \cref{prop:global-monotone-revised}, the angle sequence \(\{\vartheta_k\}\) is also nondecreasing.
Next, write the skew-Hermitian component of the gradient in the fixed two-dimensional space \(\mathcal W=\operatorname{span}_{\mathbb R}\{X_0,Y_0\}\) as $\mathrm g_k=[H,\psi_k]=x_kX_0+y_kY_0,$ and define the corresponding coefficient length by
\begin{equation*}
        \xi_k:=\sqrt{x_k^2+y_k^2}.
\end{equation*}
Thus \(\xi_k\) is simply the Euclidean length of the coordinate vector \((x_k,y_k)\) in the fixed basis \(\{X_0,Y_0\}\).  Then, \(\frac{1}{\xi_k}\mathrm g_k \) is a normalized gradient direction.

\begin{lemma}[Lower bound for the coefficient length]
\label{lem:coefficient-length-enter}
For every iterate of \cref{alg-newton-armijo}, we have
\begin{equation}
\label{eq:xi-q-enter}
        \xi_k^2
        =
        \frac{q_k(1-q_k)}{q_0(1-q_0)}.
\end{equation}
Moreover, with \(r_{\mathrm N}\) defined in \cref{eq:CN-rN-rmn-final}, if \(e_k>r_{\mathrm N}\), then
\begin{equation*}
    \xi_k\ge \frac1{32}.
\end{equation*}
\end{lemma}

\begin{proof}
By \cref{lem-X0YO}, \(X_0\) and \(Y_0\) are orthogonal and satisfy \(\|X_0\|_F^2=\|Y_0\|_F^2=2q_0(1-q_0)\). Hence, from \(\mathrm g_k=x_kX_0+y_kY_0\) and \(\xi_k^2=x_k^2+y_k^2\),
\[
\|\mathrm g_k\|_F^2
= x_k^2 \|X_0\|_F^2 +y_k^2 \|Y_0\|_F^2
        =
        2q_0(1-q_0)\xi_k^2 .
\]
Since also \(\|\mathrm g_k\|_F^2=2q_k(1-q_k)\), \cref{eq:xi-q-enter} follows.

It remains to prove the lower bound before entering the Newton region. Set \(p_0:=q_0(1-q_0)\). Since \(0<p_0\le1/4\), we have \(2\le1/\sqrt{p_0}\) and \(1/\sqrt{2p_0}\le1/\sqrt{p_0}\), and therefore
\[
        L_{\mathrm{Rie}}
        =
        2+\frac{1}{\sqrt{2p_0}}
        \le
        \frac{1}{\sqrt{p_0}}+\frac{1}{\sqrt{p_0}}
        =
        \frac{2}{\sqrt{p_0}}.
\]
Consequently, $r_{\mathrm N} = \frac{1}{128L_{\mathrm{Rie}}^2} \ge \frac{p_0}{512}.$ Also \(L_{\mathrm{Rie}}>2\), so \(r_{\mathrm N}<1/512<1/2\).

Now assume \(e_k=1-q_k>r_{\mathrm N}\). Then \(q_k<1-r_{\mathrm N}\), while monotonicity gives \(q_k\ge q_0\). Thus \(q_k\in[q_0,1-r_{\mathrm N}]\). Since the function \(q\mapsto q(1-q)\) is concave on \([0,1]\), its minimum over this interval is attained at an endpoint, and hence
\[
        q_k(1-q_k)
        \ge
        \min\{q_0(1-q_0),\,r_{\mathrm N}(1-r_{\mathrm N})\}.
\]
The first endpoint contributes \(1\) after division by \(p_0=q_0(1-q_0)\), while the second endpoint satisfies
\[
        \frac{r_{\mathrm N}(1-r_{\mathrm N})}{p_0}
        \ge
        \frac{(p_0/512)(1/2)}{p_0}
        =
        \frac1{1024}.
\]
Combining this endpoint bound with \cref{eq:xi-q-enter} yields \(\xi_k^2\ge1/1024\), and therefore \(\xi_k\ge1/32\).
\end{proof}

For a trial line search parameter \(t\), the tested tangent vector is \(t\gamma_k\mathrm g_kU_k\). Since \(\mathrm g_k=x_kX_0+y_kY_0\), we measure this trial step by its coefficient length in the fixed basis \(\{X_0,Y_0\}\):
\begin{equation*}
       R:=t\gamma_k\xi_k .
\end{equation*}
We call \(R\) the trial radius. The next lemma shows that, before entering the Newton region, the Armijo line search admits a uniform lower bound for the accepted radius.

\begin{lemma}[Lower bound for the accepted radius]
Suppose that \cref{alg-newton-armijo} is run with parameters \(0<\delta\le1\), \(0<c\le1/4\), and \(0<\rho<1\). Then there exists a constant \(R_*>0\) such that, whenever \(e_k>r_{\mathrm N}\), the accepted radius $R_k:=t_k\gamma_k\xi_k$ satisfies
\begin{equation}
\label{eq:accepted-radius-lower-enter}
        R_k\ge \rho R_* .
\end{equation}
\end{lemma}

\begin{proof}
Fix an iterate \(U_k\) of \cref{alg-newton-armijo} before the Newton region and omit the subscript \(k\). Write
\[
\begin{aligned}
        q&=f(U),\\
        \psi&=U\psi_0U^\dagger,\\
        \mathrm g&=[H,\psi]=xX_0+yY_0,\\
        \xi&=\sqrt{x^2+y^2}.
\end{aligned}
\]
Since \(e:=1-q>r_{\mathrm N}\), \cref{lem:coefficient-length-enter} gives \(\xi\ge1/32>0\). For a trial step \(t\), define its Grover-plane radius by
\[
        R:=t\gamma\xi .
\]
Equivalently, \(t\gamma=R/\xi\), so the corresponding trial point can be written as $\mathrm R_U^{(5)} \left( \frac{R}{\xi}\mathrm gU \right).$ We therefore define the radius parametrized trial objective function by
\begin{equation}
\label{eq:Q-R-def-enter}
        Q_U(R):=
        f\!\left(
        \mathrm R_U^{(5)}
        \left(
        \frac{R}{\xi}\mathrm gU
        \right)
        \right).
\end{equation}
In terms of \(R\), the Armijo condition becomes
\begin{equation}
\label{eq:small-radius-armijo-enter}
        Q_U(R)
        \ge
        q+c\,\frac{R}{\xi}\|\mathrm g\|_F^2 .
\end{equation}
We first show that this condition holds uniformly for all sufficiently small radii. More precisely, there exists a constant \(R_*>0\), independent of the iterate \(U\), such that \cref{eq:small-radius-armijo-enter} holds for every \(0\le R\le R_*\).

Since \(\mathrm R_U^{(5)}\) is a retraction, \(Q_U(0)=q\) and
\begin{equation}\label{eq-6111335}
\begin{aligned}
Q'_U(0)
&=\left\langle \operatorname{grad} f(U),
        \frac{1}{\xi}\mathrm gU\right\rangle\\
&=\left\langle \mathrm gU,
        \frac{1}{\xi}\mathrm gU\right\rangle\\
&=\frac{1}{\xi}\|\mathrm g\|_F^2 .
\end{aligned}
\end{equation}
Using \(\|\mathrm g\|_F^2=2q(1-q)\) and \(\xi^2=q(1-q)/(q_0(1-q_0))\) from \cref{lem:coefficient-length-enter}, we obtain
\begin{equation*}
        Q'_U(0)=2\sqrt{q_0(1-q_0)}\sqrt{q(1-q)} .
\end{equation*}
Set \(A:=\sqrt{q_0(1-q_0)}\sqrt{q(1-q)}\).
Since \(U\) is nonstationary, \(A>0\). Defining an  auxiliary function
\[
        \Phi_U(R):=\frac{Q_U(R)-q}{A}.
\]
Then \(\Phi_U(0)=0\) and \(\Phi_U'(0)=2\).

We now show that the functions \(\Phi_U\) belongs to a fixed compactly parametrized family. We use the same non-normalized Grover-plane basis as in \cref{thm-grad-classical}, also in \cref{thm-newton-classical}:
\[
        u:=H|\psi_0\rangle,
        \qquad
        v:=(I-H)|\psi_0\rangle .
\]
In this basis, the restrictions of \(H\) and \(\psi_0\) are the fixed \(2\times2\) matrices
\[
        E_H(0):=
        \begin{bmatrix}
        1&0\\
        0&0
        \end{bmatrix},
        \qquad
        \Psi_0:=
        \begin{bmatrix}
        q_0&1-q_0\\
        q_0&1-q_0
        \end{bmatrix}.
\]
More generally,
\[
        E_H(\theta):=
        \begin{bmatrix}
        e^{i\theta}&0\\
        0&1
        \end{bmatrix},
        \qquad
        E_{\psi_0}(\theta)=:I_2+(e^{i\theta}-1)\Psi_0 .
\]
By the Grover-plane invariance, the current state vector generated by $U$ has coordinates $|\psi\rangle=\alpha u+\beta v$ in this same basis. Hence
\begin{equation}\label{eq-6102222}
q=f(U)=\langle\psi|H|\psi\rangle=q_0|\alpha|^2,
        \qquad
        1-q=(1-q_0)|\beta|^2 .
\end{equation}
The corresponding rank-one projector \(\psi=|\psi\rangle\langle\psi|\), represented as a linear operator in the basis \(\{u,v\}\), is
\[
        [\psi]_{u,v}=
        \begin{bmatrix}
        q_0|\alpha|^2&(1-q_0)\alpha\overline{\beta}\\
        q_0\beta\overline{\alpha}&(1-q_0)|\beta|^2
        \end{bmatrix}.
\]
If we set $z:=\alpha\overline{\beta},$ then this becomes
\begin{equation}\label{eq-6102225}
[\psi]_{u,v}=
        \begin{bmatrix}
        q&(1-q_0)z\\
        q_0\overline z&1-q
        \end{bmatrix}.
\end{equation}
Thus, the real number \(q\) and the complex number \(z\) determine the matrix of \(\psi\).

We now relate \(z\) to the gradient coefficients $x$ and $y$. Since \(H\) is represented by \(\begin{bmatrix}1&0\\0&0\end{bmatrix}\), we have
\[
        [[H,\psi]]_{u,v}=
        \begin{bmatrix}
        0&(1-q_0)z\\
        -q_0\overline z&0
        \end{bmatrix}.
\]
On the other hand, from \(X_0=[H,\psi_0]\) and \(Y_0=i[H,X_0]\), the same non-normalized basis gives
\[
        [X_0]_{u,v}=
        \begin{bmatrix}
        0&1-q_0\\
        -q_0&0
        \end{bmatrix},
        \qquad
        [Y_0]_{u,v}=
        \begin{bmatrix}
        0&i(1-q_0)\\
        iq_0&0
        \end{bmatrix}.
\]
Therefore, for \(\mathrm g=[H,\psi]=xX_0+yY_0\), comparison of matrix entries gives $z=x+iy .$ Indeed, by \cref{eq-6102222},
\[
\xi=\sqrt{x^2+y^2}=|z|=|\alpha|\,|\beta|=\sqrt{\frac{q(1-q)}{q_0(1-q_0)}},
\]
which is the same result of \cref{lem:coefficient-length-enter}.
The above analysis based on the matrix representation is the main technical tool used to prove \cref{thm-grad-in-W}; see proof of \cite[Theorem 3.5]{lai2025grover} for details.

Consequently, if set real vector
\[
\begin{aligned}
        s&:=\frac{1}{\xi}(x,y)=(s_1,s_2)\in\mathbb S^1,\\
        \mathbb S^1&:=\{s\in\mathbb R^2:s_1^2+s_2^2=1\},
\end{aligned}
\]
then
\[
\begin{aligned}
        z\equiv z(q,s)
        &=\xi(s_1+is_2)\\
        &=\sqrt{\frac{q(1-q)}{q_0(1-q_0)}}(s_1+is_2).
\end{aligned}
\]
Substituting this into the matrix of \(\psi\) in \cref{eq-6102225}, we obtain the explicit representation
\[
        [\psi]_{u,v}\equiv [\psi]_{u,v}(q,s)=
        \begin{bmatrix}
        q&(1-q_0)z(q,s)\\
        -q_0 \overline{z(q,s)}&1-q
        \end{bmatrix}.
\]
This shows that, the current state \(\psi\) is completely determined by the pair \((q,s)\in (0,1) \times  \mathbb S^1\). The vector $s\in\mathbb R^2$ also determines the normalized trial generator in retraction, because
\[
        \frac{1}{\xi}\mathrm g=s_1X_0+s_2Y_0 .
\]
Therefore the trial curve in \cref{eq:Q-R-def-enter} depends on the current iterate \(U\) (i.e., current state $\psi$) only through \(q\) and \(s\).
Indeed, by \cref{thm-newton-classical}, the fixed Grover-compatible retraction gives a \(2\times2\) matrix \(M\equiv M(R,s)\), built only from the factors \(E_H(\theta)\) and \(E_{\psi_0}(\theta)\), such that the trial projector \(\psi(R)=|\psi(R)\rangle\langle\psi(R)|\) is represented by
\[
        [\psi(R)]_{u,v}=M(R,s)[\psi]_{u,v}M(R,s)^\dagger,
\]
Hence,
\[
\begin{aligned}
        Q_U(R)
        &=f(\mathrm R_U^{(5)}((R/\xi)\mathrm gU))\\
        &=\operatorname{tr}\!(H\psi(R))\\
        &=\operatorname{tr}\!\left(
        \begin{bmatrix}
        1&0\\
        0&0
        \end{bmatrix}
        [\psi(R)]_{u,v}
        \right),
\end{aligned}
\]
Thus there exists a scalar function induced by this fixed \(2\times2\) product such that
\[
        \Phi_U(R)=\widehat\Phi(R;q,s).
\]
The function \(\widehat\Phi\) is smooth in \((R,q,s)\), because all entries above are obtained from finite products of the smooth matrices \(E_H(\theta)\), \(E_{\psi_0}(\theta)\), and the smooth expression for \([\psi]_{u,v}\) in terms of \(q\) and \(s\).

It remains to make the uniformity precise. Since the iterates are monotone, \(q\ge q_0\). Since we are before the Newton region, hence \(q<1-r_{\mathrm N}\). Therefore, $q\in[q_0,1-r_{\mathrm N}].$ Fix some \(R_0\in(0,1]\). The parameter set
\[
        [0,R_0]\times[q_0,1-r_{\mathrm N}]\times\mathbb S^1
\]
is compact, and the function \((R,q,s)\mapsto \partial_R^2\widehat\Phi(R;q,s)\) is continuous on it. Hence, there exists a constant \(B_Q>0\), independent of \(q\) and \(s\), and therefore independent of the iterate \(U\), such that
\begin{equation*}
        |\Phi_U''(R)|\le B_Q,
        \qquad
        0\le R\le R_0 .
\end{equation*}
Equivalently,
\begin{equation*}
\begin{gathered}
\begin{aligned}
        |Q''_U(R)|
        &=A|\Phi_U''(R)|\\
        &\le B_Q\sqrt{q_0(1-q_0)}\sqrt{q(1-q)},
\end{aligned}\\
        0\le R\le R_0 .
\end{gathered}
\end{equation*}

Choose $R_*=\min\left\{R_0,\frac1{64},\frac2{B_Q}\right\}.$
For \(0\le R\le R_*\), Taylor's formula applied to \(\Phi_U\) gives
\[
\begin{aligned}
        \Phi_U(R)
        &=\Phi_U(0)+R\Phi_U'(0)+\frac12R^2\Phi_U''(\theta R)\\
        &=2R+\frac12R^2\Phi_U''(\theta R)\\
        &\ge 2R-\frac12B_QR^2\\
        &\ge R,
\end{aligned}
\]
where \(\theta\in(0,1)\), \(|\Phi_U''(\theta R)|\le B_Q\), and the last inequality uses \(B_QR\le 2\). Hence
\[
        Q_U(R)-q=A\Phi_U(R)\ge AR .
\]
On the other hand, by \cref{eq-6111335},
\[
        c\frac{R}{\xi}\|\mathrm g\|_F^2=cRQ'_U(0)=2cAR\le AR,
\]
where \(c\le1/4\).  Therefore every trial radius \(R\le R_*\) satisfies the
Armijo condition \cref{eq:small-radius-armijo-enter}.

It remains to prove the accepted radius lower bound \cref{eq:accepted-radius-lower-enter}. Assume \(e_k>r_{\mathrm N}\). By \cref{lem:coefficient-length-enter}, \(\xi_k\ge1/32\), while \(\gamma_k=1/\max\{\delta,2q_k-1\}\ge1\). Hence the first trial radius satisfies
\[
        R_k^{(0)}:=\gamma_k\xi_k\ge\frac1{32}>R_*,
\]
because \(R_*\le1/64\). During backtracking, the trial radii are \(R_k^{(m)}=\rho^mR_k^{(0)}\), \(m=0,1,2,\ldots\). Let \(R_k=R_k^{(m_k)}\) be the first accepted radius. If \(R_k>R_*\), then \(R_k>\rho R_*\). If \(R_k\le R_*\), then \(m_k\ge1\), and the previous trial radius \(R_k^{(m_k-1)}\) was rejected. Since every radius not exceeding \(R_*\) satisfies the Armijo condition by \cref{eq:small-radius-armijo-enter}, this rejection implies \(R_k^{(m_k-1)}>R_*\). Therefore
\[
        R_k=R_k^{(m_k)}=\rho R_k^{(m_k-1)}>\rho R_* .
\]
Thus \(R_k\ge\rho R_*\) in all cases.
\end{proof}

We now use the lower bound on the accepted radius to prove that the Grover angle progress by at least a fixed multiple of \(\sqrt{q_0(1-q_0)}\) before the Newton region is reached.

\begin{lemma}[Lower bound for the Grover angle increment]
\label{lem:angle-progress-enter}
Assume that \(0<c\le1/8\).  Define the constant
\begin{equation*}
        \omega_*
        :=
        \min\left\{
        \frac1{48},
        \frac{c\rho R_*}{4}
        \right\}.
\end{equation*}
Recall the definition of Grover angle $\vartheta_k$ in \cref{eq:grover-angle-enter}. If $e_k>r_{\mathrm N}$, then
\begin{equation}
\label{eq:angle-progress-enter}
        \vartheta_{k+1}-\vartheta_k
        \ge
        \omega_*\sqrt{q_0(1-q_0)}.
\end{equation}
\end{lemma}

\begin{proof}
Set \(\Delta_k:=\vartheta_{k+1}-\vartheta_k\) and \(u_0:=\sqrt{q_0(1-q_0)}\). Since \(q_k=\sin^2\vartheta_k\) and \(\xi_k=\sqrt{q_k(1-q_k)}/u_0\), we have
\begin{equation*}
        \sin(2\vartheta_k)=2 \sin(\vartheta_k)\cos(\vartheta_k)
        =
        2\sqrt{q_k(1-q_k)}
        =
        2\xi_k u_0.
\end{equation*}
Because \(e_k>r_{\mathrm N}\), \cref{lem:coefficient-length-enter} gives \(\xi_k\ge1/32\), and hence
\begin{equation}
\label{eq:sin2theta-lower-enter}
        \sin(2\vartheta_k)\ge \frac{u_0}{16}.
\end{equation}
Set $R_k:=t_k \gamma_k \xi_k$. The Armijo condition and the accepted radius bound \cref{eq:accepted-radius-lower-enter} imply
\begin{equation}
\label{eq:q-increase-angle-enter}
\begin{aligned}
        q_{k+1}-q_k
        &\ge
        c\frac{R_k}{\xi_k}G_k
        \\
        &\ge
        c\frac{\rho R_*}{\xi_k}\,2q_k(1-q_k)
        \\
        &=
        c\rho R_*\,u_0\sin(2\vartheta_k)
        \\
        &\ge
        \frac{c\rho R_*}{2}u_0\sin(2\vartheta_k).
\end{aligned}
\end{equation}
On the other hand, since \(q_k=\sin^2\vartheta_k\),
\begin{align*}
q_{k+1}-q_k
        &=
        \sin^2\vartheta_{k+1}-\sin^2\vartheta_k \notag\\
        &=
\sin(\vartheta_{k+1}+\vartheta_k)\sin(\vartheta_{k+1}-\vartheta_k) \notag \\
&=
\sin(2\vartheta_k+\Delta_k)\sin\Delta_k .
\end{align*}
Using \(\sin(a+b)\le \sin a+b\) for \(a,b\ge0\) and \(\sin\Delta_k\le\Delta_k\), we get
\begin{equation}
\label{eq:q-diff-angle-upper-enter}
        q_{k+1}-q_k
        \le
        \bigl(\sin(2\vartheta_k)+\Delta_k\bigr)\Delta_k .
\end{equation}

We claim that \(\Delta_k\ge\omega_*u_0\).  Suppose, for contradiction, that \(\Delta_k<\omega_*u_0\). By \cref{eq:sin2theta-lower-enter},
\[
        \Delta_k<\omega_*u_0\le16\omega_*\sin(2\vartheta_k).
\]
Since \(\omega_*\le1/48\), this gives \(\sin(2\vartheta_k)+\Delta_k<(1+16\omega_*)\sin(2\vartheta_k)\le(4/3)\sin(2\vartheta_k)\). Substituting into \cref{eq:q-diff-angle-upper-enter} and using \(\Delta_k<\omega_*u_0\), we obtain
\[
        q_{k+1}-q_k
        <
        \frac43\omega_*u_0\sin(2\vartheta_k).
\]
Because \(\omega_*\le c\rho R_*/4\), we have $\frac43\omega_* \le \frac{c\rho R_*}{3} < \frac{c\rho R_*}{2}.$ Therefore
\[
        q_{k+1}-q_k
        <
        \frac{c\rho R_*}{2}u_0\sin(2\vartheta_k),
\]
which contradicts \cref{eq:q-increase-angle-enter}. Therefore \(\Delta_k\ge\omega_*u_0\), which is exactly \cref{eq:angle-progress-enter}.
\end{proof}

\begin{proposition}[Entrance complexity for the Newton region]
\label{prop:entrance-newton-region-rmn-final}
Run \cref{alg-newton-armijo} with the \(5\)-factor retraction and fixed parameters \(0<\delta\le1/2\), \(0<c\le1/8\), and \(0<\rho<1\). Let
\begin{equation*}
        K_{\mathrm{enter}}:=\min\{k\ge0:e_k\le r_{\mathrm N}\},
\end{equation*}
where \(r_{\mathrm N}\) is defined in \cref{eq:CN-rN-rmn-final}. Then
\begin{equation*}
        K_{\mathrm{enter}}=\mathcal O\!\left(\sqrt{\frac{N}{M}}\right).
\end{equation*}
\end{proposition}

\begin{proof}
If \(e_0\le r_{\mathrm N}\), then \(K_{\mathrm{enter}}=0\), so assume \(e_0>r_{\mathrm N}\). By monotonicity, \(q_k\) is nondecreasing and hence \(e_k=1-q_k\) is nonincreasing. Therefore, by the definition of \(K_{\mathrm{enter}}\), every iterate before entering the Newton region satisfies \(e_k>r_{\mathrm N}\) for \(k=0,\ldots,K_{\mathrm{enter}}-1\). Applying \cref{lem:angle-progress-enter} gives
\[
        \vartheta_{k+1}-\vartheta_k
        \ge
        \omega_*\sqrt{q_0(1-q_0)},
        \qquad
        k=0,\ldots,K_{\mathrm{enter}}-1.
\]
Summing these inequalities yields
\[
\begin{aligned}
        \vartheta_{K_{\mathrm{enter}}}-\vartheta_0
        &=
        \sum_{k=0}^{K_{\mathrm{enter}}-1}
        \bigl(\vartheta_{k+1}-\vartheta_k\bigr)
        \\
        &\ge
        K_{\mathrm{enter}}\omega_*\sqrt{q_0(1-q_0)}.
\end{aligned}
\]
Since \(0<\vartheta_k\le\pi/2\), we obtain
\begin{equation}
\label{eq:Kenter-angle-count-enter}
        K_{\mathrm{enter}}
        \le
        \frac{\pi/2-\vartheta_0}{\omega_*\sqrt{q_0(1-q_0)}}
        \le
        \frac{\pi}{2\omega_*\sqrt{q_0(1-q_0)}}.
\end{equation}
Finally, \(q_0=M/N\), and the global assumption \(M\le N/2\) gives \(1-q_0\ge1/2\). Hence \(\sqrt{q_0(1-q_0)}\ge\sqrt{q_0/2}\), and substituting this bound into \cref{eq:Kenter-angle-count-enter} gives
\[
        K_{\mathrm{enter}}
        \le
        \frac{\pi}{2\omega_*}\sqrt{\frac{2}{q_0}}
        =
        \frac{\pi}{\sqrt2\,\omega_*}\sqrt{\frac{N}{M}}.
\]
Thus \(K_{\mathrm{enter}}=\mathcal O(\sqrt{N/M})\).
\end{proof}

\subsection{Total complexity}

We now combine the entrance estimate with the local quadratic convergence result.
The picture is simple: the global phase reaches the Newton region in \(\mathcal O(\sqrt{N/M})\) iterations, and once inside this region, the error decreases quadratically.
Thus the dependence on the final precision is only double logarithmic.

\begin{theorem}[Total complexity of \cref{alg-newton-armijo}]
\label{thm:global-rmn-complexity-final}
Run \cref{alg-newton-armijo} from \(U_0=I\) with the \(5\)-factor retraction and fixed parameters \(0<\delta\le1/2\), \(0<c\le1/8\), and \(0<\rho<1\). Let $T_\varepsilon:=\min\{k\ge0:e_k:=1-q_k \le\varepsilon\}$ be the number of iterations needed to reach success probability error at most \(\varepsilon\). Then, for \(0<\varepsilon<e^{-1}\),
\begin{equation}
\label{eq:sharp-global-bound-final}
        T_\varepsilon
        =
        \mathcal O\!\left(
        \sqrt{\frac{N}{M}}
        +
        \log\log\frac{1}{\varepsilon}
        \right).
\end{equation}
In other words, the total cost consists of an \(\mathcal O(\sqrt{N/M})\) entrance phase and an \(\mathcal O(\log\log(1/\varepsilon))\) Newton phase.

\end{theorem}

\begin{proof}
Let \(K_{\mathrm{enter}}\) be the first index such that \(e_{K_{\mathrm{enter}}}\le r_{\mathrm N}\). By \cref{prop:entrance-newton-region-rmn-final}, $K_{\mathrm{enter}} = \mathcal O\!\left(\sqrt{\frac{N}{M}}\right).$ After this point, \cref{thm:local-quadratic-revised} applies and gives
\[
        e_{K_{\mathrm{enter}}+j+1}
        \le
        C_{\mathrm N}e_{K_{\mathrm{enter}}+j}^2,
        \qquad j=0,1,2,\ldots .
\]
Set \(a_j:=C_{\mathrm N}e_{K_{\mathrm{enter}}+j}\). Since \(r_{\mathrm N}=1/(2C_{\mathrm N})\), we have \(a_0\le1/2\), and the recurrence above gives \(a_{j+1}\le a_j^2\). Hence, by induction,
\begin{equation*}
        e_{K_{\mathrm{enter}}+j}
        \le
        \frac1{C_{\mathrm N}}2^{-2^j}.
\end{equation*}
Therefore \(e_{K_{\mathrm{enter}}+j}\le\varepsilon\) once \(2^j\gtrsim \log(1/\varepsilon)\), which requires only $j=\mathcal O\!\left(\log\log\frac1\varepsilon\right)$ post RMN iterations. Combining this estimate with the entrance bound proves \cref{eq:sharp-global-bound-final}.

\end{proof}

\bibliography{mybib}

\end{document}